\documentclass[webpdf,contemporary,large]{oup-authoring-template}%

\graphicspath{{Fig/}}

\usepackage{booktabs}
\theoremstyle{thmstyleone}%
\newtheorem{theorem}{Theorem}
\newtheorem{proposition}[theorem]{Proposition}%
\theoremstyle{thmstyletwo}%
\theoremstyle{thmstylethree}%

\begin{document}

\journaltitle{Journal Title Here}
\DOI{DOI added during production}
\copyrightyear{2026}
\pubyear{2026}
\vol{XX}
\issue{x}
\access{Published: Date added during production}
\appnotes{Paper}

\firstpage{1}


\title[PI-Mamba]{PI-Mamba: Linear-Time Protein Backbone Generation via Spectrally Initialized Flow Matching}

\author[1,$\ast$]{Tianyu Wu}
\author[2]{Lin Zhu}

\address[1]{\orgdiv{Center for Biophysics and Quantitative Biology}, \orgname{University of Illinois Urbana-Champaign}, \orgaddress{\street{600 S Mathews Ave.}, \postcode{61801}, \state{IL}, \country{USA}}}
\address[2]{\orgdiv{School of Information Science}, \orgname{University of Illinois Urbana-Champaign}, \orgaddress{\street{501 E Daniel St.}, \postcode{61820}, \state{IL}, \country{USA}}}

\corresp[$\ast$]{Tianyu Wu. \href{email:tianyu16@illinois.edu}{tianyu16@illinois.edu}}

\received{Date}{0}{Year}
\revised{Date}{0}{Year}
\accepted{Date}{0}{Year}



\abstract{\textbf{Motivation:} Generative models for protein backbone design have to simultaneously ensure geometric validity, sampling efficiency, and scalability to long sequences. However, most existing approaches rely on iterative refinement, quadratic attention mechanisms, or post-hoc geometry correction, leading to a persistent trade-off between computational efficiency and structural fidelity.\\
\textbf{Results:} We present Physics-Informed Mamba (PI-Mamba), a generative model that enforces exact local covalent geometry by construction while enabling linear-time inference. PI-Mamba integrates a differentiable constraint-enforcement operator into a flow-matching framework and couples it with a Mamba-based state-space architecture. To improve optimisation stability and backbone realism, we introduce a spectral initialisation derived from the Rouse polymer model and an auxiliary cis-proline awareness head. Across benchmark tasks, PI-Mamba achieves 0.0\% local geometry violations and high designability (scTM = 0.91 ± 0.03, n = 100), while scaling to proteins exceeding 2,000 residues on a single A5000 GPU (24 GB).\\}
\keywords[Abbreviations]{PI-Mamba, SSM, FM, FAPE, scTM, NeRF, SE(3)}
\keywords[Availability:]{Cath4.2 is available at \url{https://www.cathdb.info/wiki/doku/?id=data:index}. Distilled Dataset and Code are available in xxx}




\maketitle


\section{Introduction}
Protein structure generation is fundamental to protein design and synthetic biology. Recent generative models, diffusion-based methods \cite{Watson2023,yimFastProteinBackbone2023} and large auto-regressive architectures \cite{ingrahamGenerativeModelsGraphbased2019a}, have achieved impressive fidelity on short to medium-length proteins. However, scaling generation to long chains remains challenging due to a fundamental tension between global dependency modeling, computational efficiency, and strict geometric validity. Most existing methods enforce physical constraints via soft penalties or post-hoc refinement, which degrade as sequence length increases, leading to accumulated bond violations and reliance on expensive downstream relaxation \cite{Jumper2021}.

Structured state space models (SSMs), particularly selective variants such as Mamba \cite{guMambaLineartimeSequence2023}, offer linear-time inference and strong long-context capabilities, making them attractive alternatives to attention-based models that scale quadratically \cite{vaswaniAttentionAllYou2023}. However, standard SSMs are agnostic to physical constraints and lack inductive biases for polymeric systems, making them ill-suited for protein backbone generation, where every sample must lie on a highly constrained kinematic manifold.
We introduce \textbf{PI-Mamba}, a generative model for protein backbone design that can (i) guarantee zero bond-length and bond-angle violations without post-hoc relaxation, (ii) scale to proteins exceeding 2{,}000 residues on a single NVIDIA RTX A5000 GPU (24\,GB) with linear-time and linear-memory inference, and (iii) achieve high designability (self-consistency TM-score $0.91 \pm 0.03$, $n=100$) \cite{zhangScoringFunctionAutomated2004}. Unlike diffusion-based methods that require quadratic attention and iterative geometry correction, PI-Mamba produces geometrically valid backbones by construction at every generation step, enabling high-throughput exploration of protein design spaces impractical for existing approaches.

\section{Related Work}

\textbf{Generative Protein Design} Protein backbone generation is increasingly driven by diffusion and flow-matching models. Representive models include RFdiffusion \cite{Watson2023} fine-tuned on RoseTTAFold \cite{Baek2021} and FrameDiff \cite{Yim2023} diffused on $SE(3)$ frames; while flow matching \cite{lipmanFlowMatchingGenerative2022, albergoBuildingNormalizingFlows2022} offers alternatives to score-based diffusion, with FrameFlow \cite{yimFastProteinBackbone2023} and OriginFlow \cite{yanRobustReliableNovo2025} exploring $SE(3)$ interpolants. Our approach couples manifold-aware generation with linear-time Mamba.


\textbf{State Space Models.} S4 \cite{Gu2022} and Mamba \cite{guMambaLineartimeSequence2023} achieve $O(L)$ complexity. In biology, Mamba has been applied to DNA \cite{Schiff2024}, RNA \cite{wangProteinConformationGeneration2024}, and protein modelling \cite{Sgarbossa2024}.  Mamba's selective updates suit proteins where sparse residues interact at each position, enabling exploration of conformational space at $L>1000$. To our knowledge, PI-Mamba is among the first to combine SSM architectures with $SE(3)$ frame-based generative protein design.

\section{Materials and Methods}

\label{sec:methods}

PI-Mamba generates protein backbones with two goals: (i) linear-time sampling in sequence length and
(ii) guaranteed local covalent geometry without post-hoc relaxation. We first describe the overall
pipeline and then introduce notation used throughout the Methods. Technical derivations are provided in Appendix~\ref{sec:math_deriv}; implementation details and hyperparameters in Appendix~\ref{sec:hyperparams}.

\subsection{Overview}
\label{sec:methods_overview}
PI-Mamba (Fig.~\ref{fig:architecture}) is a physics-informed generative framework that adapts Mamba-style state space models to efficient 3D protein backbone synthesis under geometric and physical constraints. The model takes as input a sequence of length $L$ together with an $SE(3)$-aware noise prior, and generates structured backbones by integrating a learned time-dependent vector field over a series of small steps. The state is represented as residue-local rigid frames, enabling equivariant perturbations of rigid-body degrees of freedom while avoiding expensive all-atom refinement during sampling.

A central design choice is to enforce covalent-geometry constraints \emph{during} generation rather than as post-processing. Specifically, a differentiable kinematic constraint enforcement operator periodically retracts the intermediate $C_\alpha$ trace to maintain ideal bond lengths and angles at every generation step, and full backbone atoms (N, $C_\alpha$, C, O) are reconstructed at output using a deterministic internal-coordinate procedure (NeRF) with fixed bond lengths, bond angles, and peptide planarity. This makes constraint enforcement part of the generative trajectory rather than a post-hoc correction.

Global consistency is obtained with a bidirectional Mamba backbone that achieves $O(L)$ inference complexity via SSM recurrence rather than quadratic attention. A key contribution is \emph{physics-informed spectral initialization}: the SSM state transition matrices are initialized using the Rouse polymer model \cite{rouseTheoryLinearViscoelastic1953,doiTheoryPolymerDynamics2013}, encoding meaningful long-range correlations from the outset. This provides an interpretable mode hierarchy---low-frequency modes capture global topology while high-frequency modes capture local fluctuations---stabilizing generation and preventing geometric pathologies such as molten-globule collapse. The Rouse initialization also induces structure-dependent memory timescales, where helices relax faster than loops (Appendix~\ref{sec:ipa_details} describes the optional full IPA variant).

\begin{figure}[t]
\begin{center}
\includegraphics[width=0.95\columnwidth]{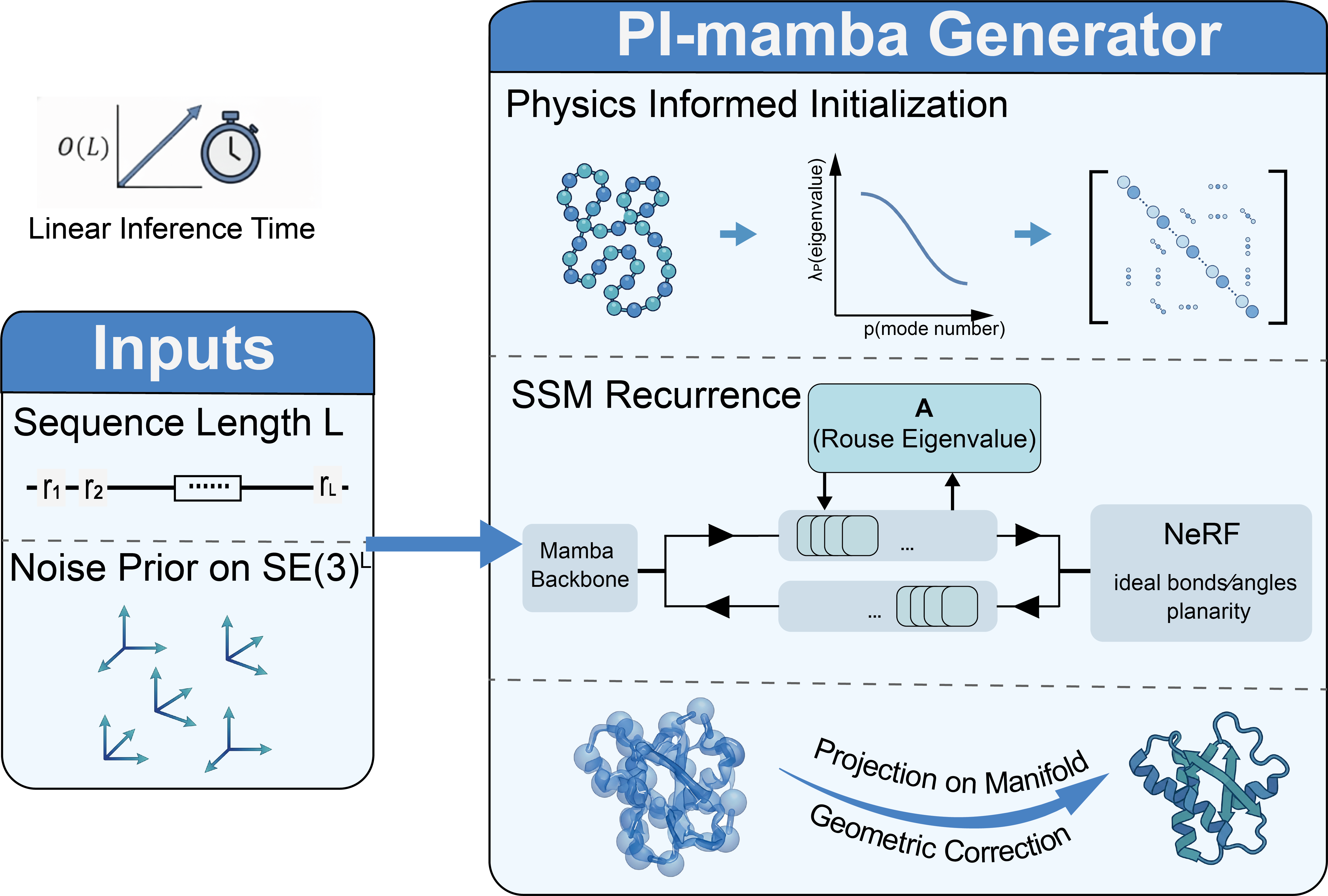}
\caption{\textbf{Architecture of PI-Mamba for scalable protein design.} The framework processes a sequence length L and a noise prior on SE(3)L through a unified generator optimized for O(L) inference. The pipeline consists of three integrated lanes: (1) \textbf{Physics-informed spectral initialization}, which uses the Rouse spectrum $\lambda_p=4sin^2(p\pi /2L)$ to initialize SSM decay rates $A_p=exp(-\lambda_p\Delta t/\tau(x))$, ensuring global topology is captured by low modes and local details by high modes; (2) a \textbf{Bidirectional Mamba backbone} that predicts tangent velocities $\xi_\theta(t,T)\in SE(3)^L$ for flow matching, enabling iterative updates $T_{t+\Delta t}=T_t\cdot exp(\Delta t\xi_\theta)$ and backbone reconstruction via NeRF; and (3) \textbf{Constraint enforcement}, where periodic $C_\alpha$ retraction $\Pi_K$ maintains ideal 3.80\AA bond lengths and peptide planarity.}
\label{fig:architecture}
\end{center}
\end{figure}
\subsection{Notation}
\label{sec:methods_notation}

Let $L$ denote the protein length. We represent a backbone as a sequence of residue-local rigid frames
\begin{equation}
T \;=\; \{T_i\}_{i=1}^L \in SE(3)^L, \qquad T_i=(R_i,\mathbf{t}_i)
\end{equation}
where $\mathbf{t}_i\in\mathbb{R}^3$ is the $C_\alpha$ coordinate and $R_i\in SO(3)$ encodes the orientation of
the local residue frame. We write $\mathbf{X}(T)\in\mathbb{R}^{3L}$ for the $C_\alpha$ trace extracted from $T$.

PI-Mamba defines a time-dependent vector field in the Lie algebra,
\begin{equation}
\xi_\theta(t,T_t)\in \mathfrak{se}(3)^L
\end{equation}
and generates samples by numerically integrating from $t=0$ to $t=1$. We use $S$ to denote the number of
integration steps (to avoid overloading $T$ as both frames and step count), with step size $\Delta t = 1/S$.

Constraint enforcement acts on $C_\alpha$ traces via a deterministic operator $\Pi_K$.
We apply $\Pi_K$ every $k$ integration steps (projection frequency $k$). Final backbone atoms are obtained
by deterministic reconstruction with ideal internal coordinates; we denote the final output as
$\widehat{\mathbf{Y}}(T_1)$. A summary of all notation is provided in Appendix~\ref{sec:notation}. Unless stated otherwise, PI-Mamba inference time includes both in-loop retraction
and final reconstruction.

Designability is evaluated using self-consistency TM-score (scTM) \cite{zhangScoringFunctionAutomated2004}, computed by inverse-folding generated backbones to sequences and re-predicting structures, then averaging TM-score between the re-predicted structure and the original backbone (full protocol in Appendix~\ref{sec:reproducibility}).

\subsection{Geometric flow matching on $SE(3)^L$}
\label{sec:methods_fm}
PI-Mamba models backbone generation as transport on the product manifold $\mathcal{M}=SE(3)^L$.
Given an initial noisy state $T_0\sim p_0$ and a learned time-dependent vector field
$\xi_\theta(t,T_t)\in \mathfrak{se}(3)^L$, samples are generated by integrating
\begin{equation}
\frac{dT_t}{dt} \;=\; T_t \cdot \xi_\theta(t,T_t),
\qquad t\in[0,1]
\label{eq:se3_ode}
\end{equation}
where $\cdot$ denotes the left-trivialized action of the Lie algebra on the group.
We use a first-order Lie-group integrator with step size $\Delta t=1/S$:
\begin{equation}
T_{t+\Delta t} \;=\; T_t \cdot \exp\!\big(\Delta t\,\xi_\theta(t,T_t)\big)
\label{eq:lie_euler}
\end{equation}
which is efficient in the small-step regime and preserves equivariance by construction (error bound in Appendix~\ref{app:proofs}).

Training uses flow matching to regress the vector field toward a prescribed probability path between the prior distribution $p_0$ and the data distribution $p_1$. In practice, this reduces to sampling $t\sim\mathcal{U}[0,1]$, drawing an endpoint example from the data distribution, constructing an intermediate state at time $t$ using the chosen path, and minimizing a mean-squared error between $\xi_\theta(t,T_t)$ and the target velocity (details of the chosen path and tangent-space implementation are provided in Appendix~\ref{sec:math_deriv}). Because these derivations are standard for geometric flow matching,
We focus below on the two components that differentiate PI-Mamba: (i) physics-informed initialisation of the sequence model that predicts $\xi_\theta$, and (ii) explicit constraint enforcement inside the
integration loop.
\subsection{Physics-informed spectral initialisation of the Mamba backbone}
\label{sec:methods_rouse_mamba}

The core sequence model is a bidirectional Mamba state space network that predicts the per-residue twist velocities required by Eq.~\ref{eq:lie_euler}. A key contribution is that we initialise the state space dynamics using polymer physics rather than generic random or learned spectra. We draw inspiration from the Rouse model, in which a polymer chain is represented as beads connected by harmonic springs. The chain connectivity induces a Laplacian spectrum with eigenvalues
\begin{equation}
\lambda_p \;=\; 4\sin^2\!\left(\frac{p\pi}{2L}\right)
\qquad p=0,\ldots,L-1
\label{eq:rouse_lambda}
\end{equation}
which correspond to relaxation modes ordered from global (low $p$) to local (high $p$).

Motivated by the correspondence between discretised Rouse relaxation and diagonal structured state space models, we initialise the Mamba state transition in mode space as
\begin{equation}
A_p \;=\; \exp\!\left(-\lambda_p\,\frac{\Delta t}{\tau(\mathbf{x})}\right)
\label{eq:rouse_init}
\end{equation}

where $\Delta t$ is the (learned) discretization step and $\tau(\mathbf{x})$ is an input-dependent relaxation time predicted from local features. This imposes an inductive bias in which low-frequency modes persist longer and carry global topology, while high-frequency modes decay quickly and capture
local fluctuations. The learned $\tau(\mathbf{x})$ modulates this hierarchy in a structure-dependent manner, allowing flexible regions to maintain longer-range correlations than rigid secondary-structure
segments.

To improve long-range consistency without quadratic attention, we process sequences bidirectionally and average forward and reverse scans. This yields $O(L)$ time and memory scaling in the sequence backbone. We quantify how strongly the internal representations align with the Rouse basis using the Rouse mode concentration metric $\mathcal{C}_k$ (definition and computation in Supplemental File), and we report its correlation with downstream designability in Results.

\subsection{Constraint enforcement during generation}
\label{sec:methods_constraints}

Physically valid backbones lie on a constrained subset of $\mathbb{R}^{3L}$ that enforces adjacent
bond-length scale and local covalent geometry. Rather than relying on soft penalties or post-hoc
relaxation, PI-Mamba enforces constraints explicitly during sampling.

Let $\mathbf{X}(T)\in\mathbb{R}^{3L}$ denote the $C_\alpha$ trace extracted from frames $T$.
We define a deterministic retraction operator $\Pi_K$ that updates the trace sequentially to enforce
the target adjacent distances:
\begin{equation}
\mathbf{x}_{i+1} \leftarrow \mathbf{x}_i + d_i
\frac{\mathbf{x}_{i+1}-\mathbf{x}_i}{\|\mathbf{x}_{i+1}-\mathbf{x}_i\|},
\qquad i=1,\ldots,L-1
\label{eq:ca_retract}
\end{equation}
where $d_i=3.80$~\AA{} for trans peptide bonds and $d_i=2.96$~\AA{} for cis-proline transitions.
This step prevents integration drift in the adjacent distance scale and stabilizes the generative
trajectory. We apply $\Pi_K$ every $k$ integration steps (projection frequency $k$).

At output, we reconstruct full backbone atoms (N, $C_\alpha$, C, O) using a deterministic internal
coordinate procedure with fixed ideal bond lengths, bond angles, and peptide planarity. This yields
zero bond-length and peptide-planarity violations by construction. Importantly, this reconstruction is
part of PI-Mamba’s reported inference cost, rather than an external refinement step. The scope of what
is \emph{guaranteed} by construction versus what is \emph{encouraged} by learning is summarized in SI,
together with formal statements of the guarantees.
\subsection{Training objective and auxiliary geometric supervision}
\label{sec:methods_objective}

PI-Mamba is trained to match the target flow while producing structures that are close to the constraint manifold even before projection. The primary training signal is the flow matching regression loss on $\xi_\theta(t,T_t)$ (SI). To improve geometric fidelity of intermediate states and reduce the
gap between raw and constrained outputs, we add auxiliary structure losses evaluated on coordinates obtained from the current frames (and after retraction when applicable).

We use Frame Aligned Point Error (FAPE) to supervise local structure in residue-centered frames:
\begin{equation}
\mathcal{L}_{\mathrm{FAPE}} =
\frac{1}{L} \sum_{i,j} \sum_{a}
\min\!\left(d_{\mathrm{clamp}},
\left\|T_i^{-1}T_j\mathbf{x}_a-(T_i^{gt})^{-1}T_j^{gt}\mathbf{x}_a\right\|\right)
\label{eq:fape}
\end{equation}
where $d_{\mathrm{clamp}}=10$~\AA \space and $a \in \{N,C_\alpha,C,O\}$. In addition, we include lightweight geometric penalties that encourage realistic local backbone statistics (e.g., bond and dihedral preferences); explicit definitions and weights are provided in the SI. The overall training objective is
\begin{equation}
\begin{split}
\mathcal{L}_{\mathrm{aux}} \;&=\; \mathcal{L}_{\mathrm{FAPE}} + \mathcal{L}_{\mathrm{bond}}
+ \mathcal{L}_{\mathrm{Rama}} + \mathcal{L}_{\mathrm{HB}} \\
\mathcal{L}_{\mathrm{total}} \;&=\; \mathcal{L}_{\mathrm{FM}} \;+\; \lambda_{\mathrm{aux}}\mathcal{L}_{\mathrm{aux}}
\end{split}
\label{eq:total_loss}
\end{equation}
Here, $\mathcal{L}_{\mathrm{FM}}$ is the flow matching regression loss that trains the network to predict the target velocity field on $SE(3)^L$ (Appendix~\ref{sec:math_deriv}). The auxiliary weight $\lambda_{\mathrm{aux}}=1.0$ balances the primary flow objective against the geometric terms. $\mathcal{L}_{\mathrm{FAPE}}$ (Eq.~\ref{eq:fape}) supervises pairwise distances between backbone atoms in residue-local frames, penalizing deviations from the ground-truth structure up to a clamped radius. $\mathcal{L}_{\mathrm{bond}}$ penalizes deviations of covalent bond lengths (N--$C_\alpha$, $C_\alpha$--C, C--O, and C--N$_{\mathrm{next}}$) from their ideal values, encouraging local stereochemical accuracy in the predicted coordinates. $\mathcal{L}_{\mathrm{Rama}}$ discourages backbone dihedral angles $(\phi,\psi)$ that fall outside favored regions of the Ramachandran plot, using a kernel density estimate derived from high-resolution crystal structures. $\mathcal{L}_{\mathrm{HB}}$ promotes formation of backbone hydrogen bonds (N--H$\cdots$O{=}C) by rewarding donor--acceptor distances near $2.9$\,\AA{} with appropriate angular geometry. Together, these terms reduce reliance on projection as a ``repair'' mechanism and empirically lead to small differences between raw and constrained outputs.


\subsection{Training data, distillation, and length curriculum}
\label{sec:methods_data_distill_curriculum}

We train on a hybrid dataset that combines real protein domains with a distilled synthetic corpus. Real structures are drawn from CATH~4.2 filtered at 40\% sequence identity\cite{Sillitoe2019}, with fixed training, validation, and test splits (Appendix~\ref{sec:reproducibility}). To improve coverage of high-designability backbones without incurring
quadratic-cost teacher inference at generation time, we additionally distill from a stronger teacher model by sampling a large set of candidate backbones and filtering for designability using scTM. PI-Mamba is trained on the mixture of real and distilled examples; ablations in Results quantify the contribution of distillation beyond CATH-only training.

Optimization uses AdamW with cosine decay and linear warmup (Appendix~\ref{sec:hyperparams}). Because pairwise structural supervision scales with $O(L^2)$ during training, we employ a progressive length curriculum: we start from shorter sequences and increase the maximum length to $L_{\max}=500$ over the first 20k steps. This
stabilizes early training and improves sample quality at longer lengths. Data augmentation includes random global $SE(3)$ transformations and residue masking (Appendix~\ref{sec:reproducibility}). Hyperparameters are selected based on validation FAPE and scTM on held-out CATH domains.

\subsection{Sampling procedure}
\label{sec:methods_sampling}

At inference time, PI-Mamba generates samples by integrating Eq.~\ref{eq:se3_ode} using the Lie-group update in Eq.~\ref{eq:lie_euler} for $S$ steps. Every $k$ steps we apply the $C_\alpha$ retraction operator in Eq.~\ref{eq:ca_retract} to maintain the correct adjacent distance scale. After reaching $t=1$, we reconstruct full backbone atoms deterministically with ideal internal coordinates. Unless stated otherwise, all reported inference times include retraction and reconstruction. We use $S=100$ by default and optionally increase $S$ for higher fidelity, as reported in Results. Integration details (solver, NFEs, and comparison to baselines) are provided in Appendix~\ref{sec:ode_details}.

\begin{figure}[t]
\begin{center}
\includegraphics[width=0.95\columnwidth]{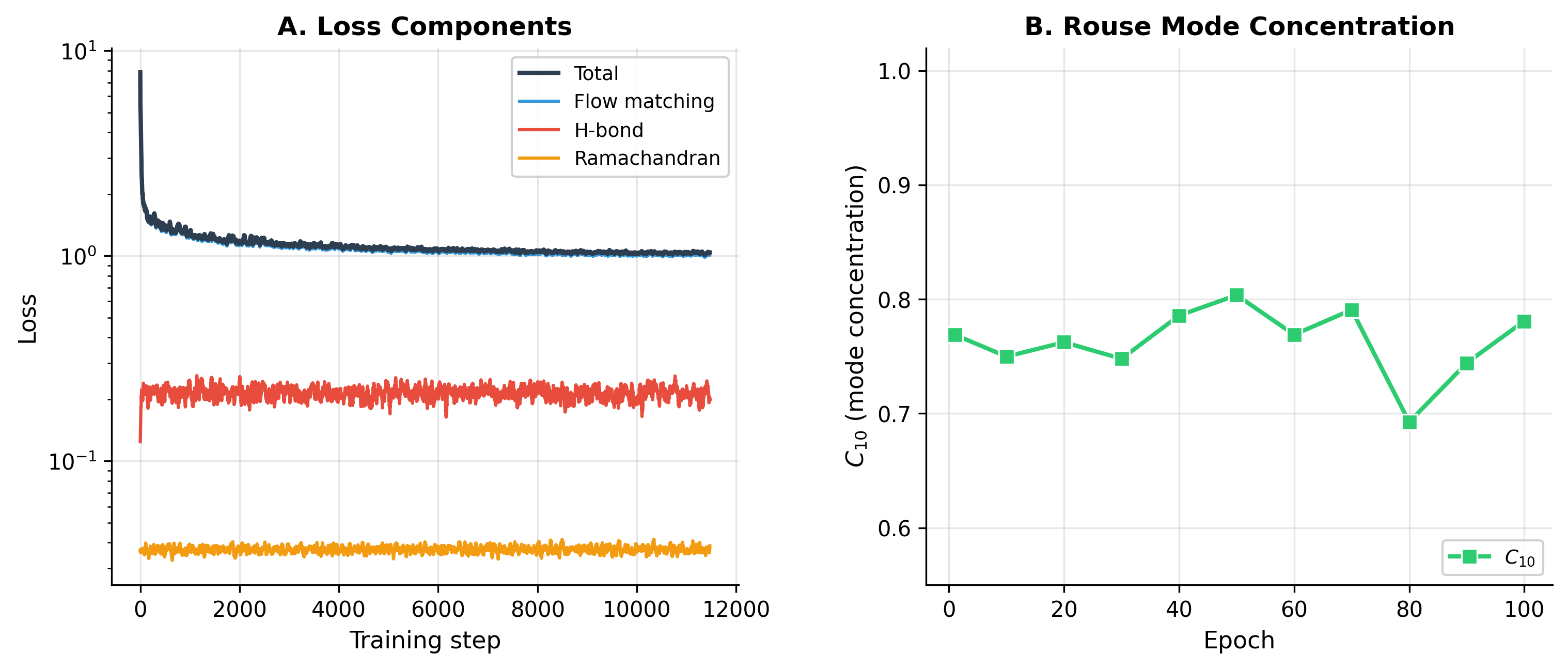}
\caption{Training dynamics of PI-Mamba. (A) Loss components (log scale): total loss, flow matching, hydrogen-bond, and Ramachandran losses all converge stably. (B) Rouse mode concentration $\mathcal{C}_{10}$ evaluated periodically during training, measuring the fraction of hidden-state variance captured by the lowest 10 Rouse modes.}
\label{fig:training_curves}
\end{center}
\end{figure}


\section{Results}
\label{sec:results}

We evaluate PI-Mamba on unconditional protein backbone generation with sequence length $L=100$, drawing $n=100$ samples per method. The representative diffusion- and flow-based baselines run with their official implementations and default settings. To fairly assess geometric fidelity, we report validity both \emph{before} and \emph{after} each method's recommended correction step (projection for PI-Mamba; post-hoc refinement for baselines). We first evaluate overall designability and validity (Section~\ref{sec:results_mainbench}--\ref{sec:results_validity}), then demonstrate scalability to long sequences (Section~\ref{sec:results_scaling}), and finally examine the mechanistic basis for these results (Section~\ref{sec:emergent_physics}). Full protocol details are provided in Appendix~\ref{sec:reproducibility}.

\subsection{Benchmark comparison on CATH backbones}
\label{sec:results_mainbench}

Table~\ref{tab:results} summarizes the primary benchmark against FrameDiff \cite{Yim2023}, FrameFlow \cite{yimFastProteinBackbone2023}, Genie2 \cite{linOutManyOne2024}, Chroma \cite{Ingraham2023}, RFdiffusion \cite{Watson2023}, Proteus \cite{Proteus2024}, and Proteina \cite{Zhang2024}. PI-Mamba obtains scTM $0.91\pm0.03$ at $L=100$ with a rounded final violation rate of $0.0\%$ and generation time of $2.3$\,s/sample, the fastest among all methods evaluated. Critically, PI-Mamba requires no post-hoc refinement (e.g., Rosetta FastRelax or OpenMM minimization) to reach zero violations, whereas all baselines depend on such correction steps, incurring substantial additional runtime. PI-Mamba is also the only method whose raw network output already exhibits near-zero violations ($0.1\%$), indicating that geometric fidelity is learned rather than imposed externally. As a sanity control, applying the projector to pure noise produces valid but non-designable structures (scTM $\approx 0$), confirming that validity alone is insufficient for designability.

\begin{table}[t]
\caption{Benchmark comparison at $L=100$ with $n=100$ samples. \textbf{Raw Viols}: \% bond violations before projection/refinement. \textbf{Final Viols}: \% after projection (PI-Mamba) or refinement (baselines). \textbf{scTM}: self-consistency TM-score (mean$\pm$SD when available). \textbf{Div}: mean pairwise RMSD (definition in Supplemental File). \textbf{Total Time}: end-to-end runtime including projection or refinement.}
\label{tab:results}
\centering
\small
\resizebox{\columnwidth}{!}{
\begin{tabular}{lccccc}
\toprule
Model & Raw Viols $\downarrow$ & Final Viols $\downarrow$ & scTM $\uparrow$ & Div $\uparrow$ & Total Time $\downarrow$ \\
\midrule
FrameDiff & 12.1\% & 0.2\% & 0.58 & 12.8\,\text{\AA} & 72.0s \\
FrameFlow & 3.2\% & 0.0\% & 0.70 & 12.2\,\text{\AA} & 22.2s \\
Genie2 & 2.1\% & 0.0\% & 0.93 & 14.4\,\text{\AA} & 60.0s \\
Chroma & 4.5\% & 0.0\% & 0.60 & 17.1\,\text{\AA} & 73.0s \\
RFdiffusion & 8.2\% & 0.0\% & 0.97 & 8.9\,\text{\AA} & 37.2s \\
Proteus & 5.4\% & 0.0\% & 0.86 & 11.2\,\text{\AA} & 13.5s \\
Proteina & 3.8\% & 0.0\% & 0.92 & 12.5\,\text{\AA} & 20.8s \\
\textit{Control: Noise + Proj.} & -- & 0.0\% & 0.00 & 28.3\,\text{\AA} & $<$1s \\
\midrule
PI-Mamba & \textbf{0.1}\% & 0.0\% & $0.91\pm0.03$ & 10.2\,\text{\AA} & \textbf{9.4s} \\
\bottomrule
\end{tabular}%
}
\raggedright\footnotesize
\textbf{Notes.} ``0.0\%'' denotes a rounded rate $<0.05\%$ under the stated threshold. All generation times measured on a single A5000 (24\,GB) with batch size 1. $^\dagger$Chroma scTM and diversity cited from \citet{boseSE3stochasticFlowMatching2023}; generation time not directly comparable. Baselines that apply post-hoc refinement (e.g., Rosetta FastRelax for RFdiffusion, OpenMM for FrameDiff) incur additional cost not shown. PI-Mamba requires no post-hoc refinement. Additional baseline protocol details are provided in Appendix~\ref{sec:reproducibility}.

\end{table}

\subsection{Geometric validity and stereochemical consistency}
\label{sec:results_validity}

The benchmark above shows that PI-Mamba reaches low final violation rates. We now examine this geometric fidelity in detail by comparing against an unconstrained Euclidean baseline (Table~\ref{tab:validity_breakdown}). PI-Mamba attains rounded final bond and angle violation rates of $0.0\%$ under the evaluation thresholds, and achieves near-perfect agreement for cis/trans proline states. Figure~\ref{fig:diversity} illustrates that PI-Mamba produces a narrow $C_\alpha$--$C_\alpha$ bond-length distribution after projection, whereas unconstrained baselines produce broad distance distributions and require substantial correction. When the same projector is applied post hoc to unconstrained outputs, the resulting structural shift is large (mean RMSD increase $>2.0$\,\AA), consistent with the projector correcting geometric errors that are not learned by those models.

\begin{table}[t]
\caption{Geometric validity breakdown (rounded). Rates computed over $n=100$ generated backbones at $L=100$.}
\label{tab:validity_breakdown}
\centering
\small
\resizebox{\columnwidth}{!}{%
\begin{tabular}{lcc}
\toprule
Metric & Euclidean baseline & PI-Mamba (Final) \\
\midrule
Bond violations ($\lvert d-3.8\rvert>0.5\,\text{\AA}$) & 76.3\% & 0.0\% \\
Angle violations ($>10^\circ$) & 3.2\% & 0.0\% \\
Cis-proline recall ($\omega$ agreement) & 35.0\% & 99.2\% \\
Trans-proline recall ($\omega$ agreement) & 92.1\% & 99.8\% \\
\bottomrule
\end{tabular}}
\end{table}

\begin{figure}[t]
\vskip 0.2in
\begin{center}
\includegraphics[width=0.58\columnwidth]{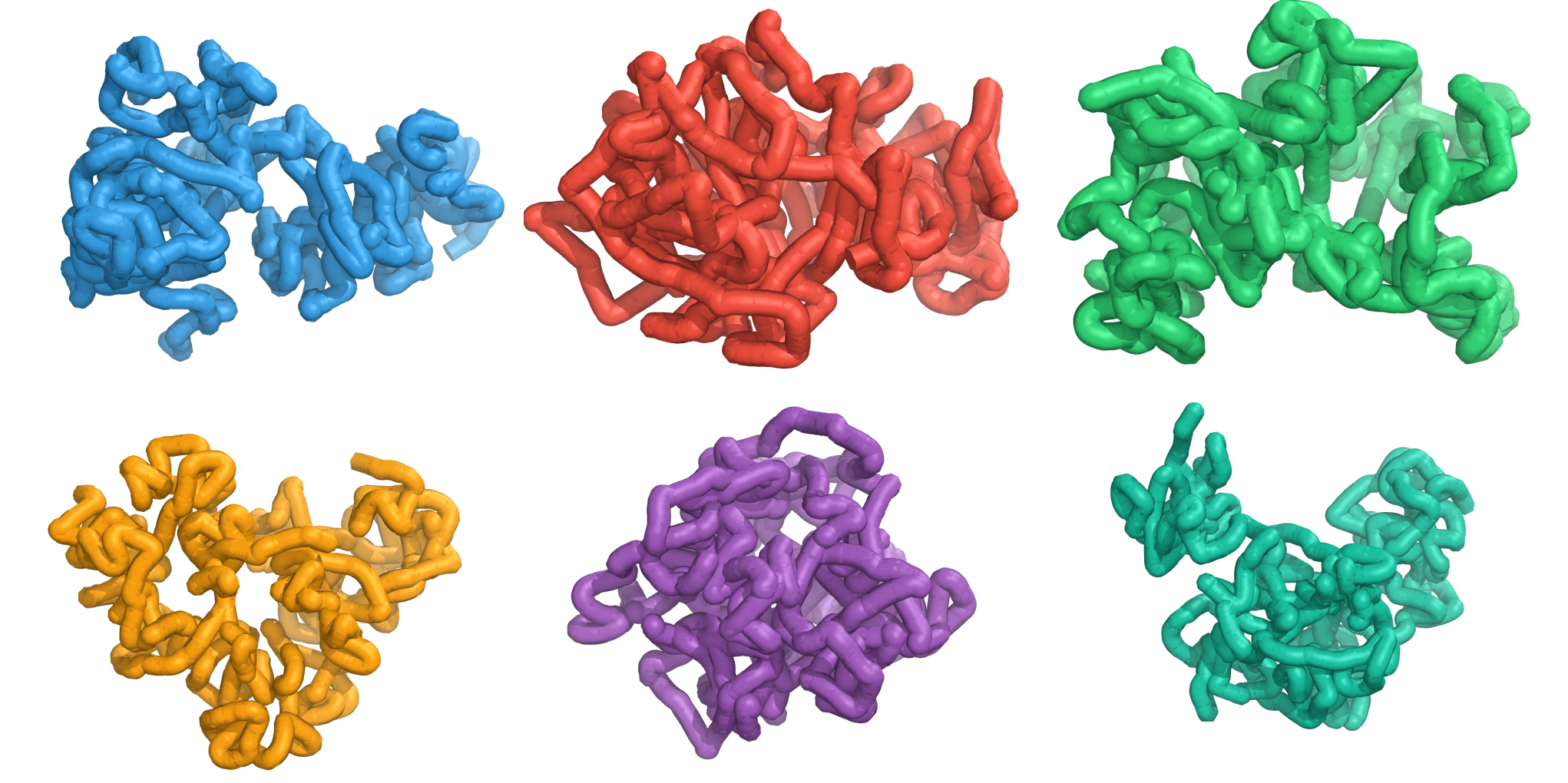}
\includegraphics[width=0.40\columnwidth]{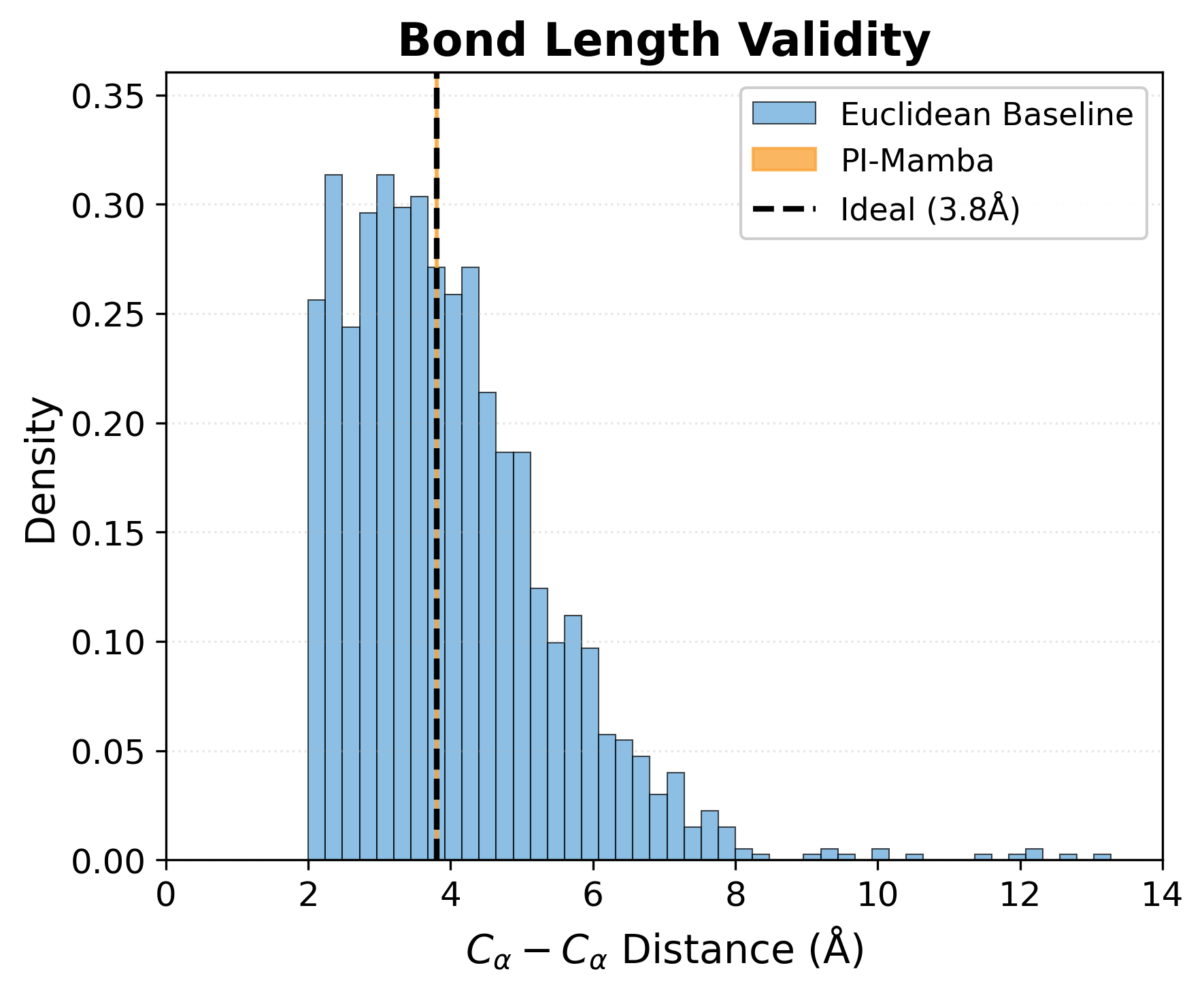}
\caption{Diversity and validity of generated backbones ($L=200$). (Left) Overlay of $6$ random samples rendered in PyMOL \cite{PyMOL} (white/black spheres: N-/C-termini). (Right) $C_\alpha$--$C_\alpha$ bond-length distribution. PI-Mamba exhibits a narrow peak near $3.8$\,\AA\ after projection.}
\label{fig:diversity}
\end{center}
\vskip -0.2in
\end{figure}

To isolate learning from projection, we also evaluate PI-Mamba \emph{before} applying projection. The raw network output exhibits a bond-violation rate of $0.1\%$ (Table~\ref{tab:results}), indicating that training drives samples close to the constraint manifold. In contrast, multiple baselines exhibit substantially higher raw violation rates and rely on refinement to achieve their final validity (Table~\ref{tab:results}).

\subsection{Comparison of Computational efficiency and scalability}
\label{sec:results_scaling}

Having established geometric validity at $L=100$, we next ask whether these properties hold at longer sequences. We benchmark inference time and peak GPU memory across sequence lengths on a single NVIDIA RTX A5000 (24\,GB) (Figure~\ref{fig:timing}). PI-Mamba exhibits near-linear scaling in both time and memory:
  at $L{=}100$, inference takes 2.3\,s per sample using 0.21\,GB; at $L{=}500$, 10.5\,s using 0.26\,GB; and at $L{=}1000$, 21.3\,s using only 0.41\,GB. In contrast, all diffusion-based baselines exhibit
   superlinear scaling and encounter out-of-memory failures on the same hardware: FrameDiff, Genie2, and Proteus fail beyond $L{\approx}500$, while Proteina cannot exceed $L{=}300$. Only FrameFlow and
  RFdiffusion reach $L{=}1000$, but at substantially higher cost---FrameFlow requires 59.7\,s and 9.6\,GB, while RFdiffusion requires 901\,s and 14.3\,GB( VRAM drop is due to the compression method they use after reach VRAM upper limit). At $L{=}500$, PI-Mamba achieves speedups of
  $21\times$ over RFdiffusion, $27\times$ over Proteus, $37\times$ over FrameDiff, and $2\times$ over FrameFlow, while using $50{-}80\times$ less GPU memory than most baselines. These results confirm
  that the Mamba backbone's linear-time sequence modeling translates directly into practical scalability advantages for protein generation.

\begin{figure}[h]
\centering
\includegraphics[width=\columnwidth]{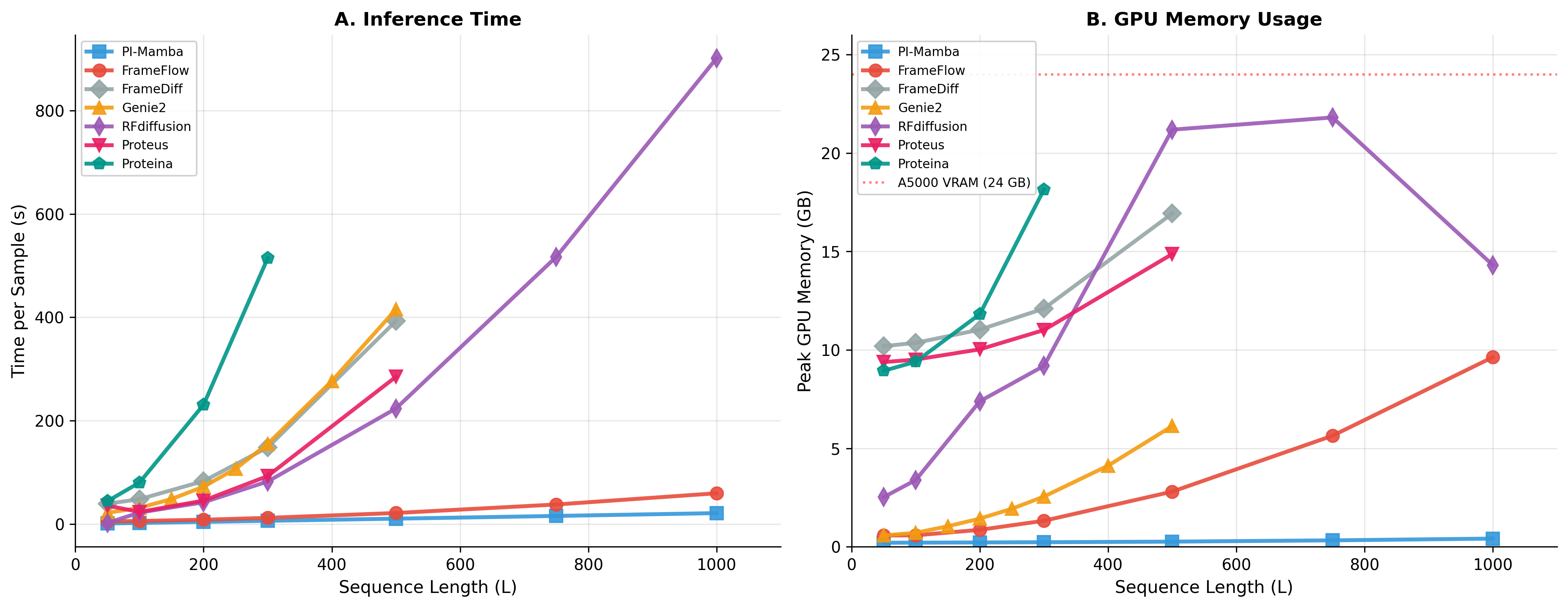}
\caption{Computational efficiency comparison on a single NVIDIA RTX A5000 (24\,GB). \textbf{(A)} Inference time per sample as a function of sequence length. PI-Mamba scales near-linearly, while
diffusion baselines exhibit superlinear growth. \textbf{(B)} Peak GPU memory usage. PI-Mamba remains below 0.5\,GB up to $L{=}1000$, whereas most baselines exceed 10\,GB and encounter OOM failures
beyond $L{\approx}500$. Endpoints indicate the maximum length each method can handle before OOM on this hardware: FrameDiff, Genie2, and Proteus fail at $L{>}500$; Proteina at $L{>}300$; only
FrameFlow and RFdiffusion reach $L{=}1000$.}
\label{fig:timing}
\end{figure}


\subsection{Emergent polymer statistics and mechanistic evidence}
\label{sec:emergent_physics}

Beyond benchmark metrics, PI-Mamba exhibits physically interpretable statistics consistent with its Rouse-inspired spectral initialization. Figure~\ref{fig:physics_app} shows that the hidden-state amplitudes follow the theoretical $1/p^2$ Rouse mode decay (panel A), the radius of gyration scales as $R_g \propto L^{0.48}$, consistent with Rouse theory (panel B), and the end-to-end distance distribution matches the Gaussian chain prediction (panel C). The learned relaxation parameter $\tau(x)$ also stratifies across secondary-structure classes, with $\tau_{\text{helix}} < \tau_{\text{sheet}} < \tau_{\text{loop}}$ (one-way ANOVA, $p<0.01$; Appendix~\ref{app:physics}). Together, these results support that the spectral inductive bias is reflected in learned internal dynamics rather than solely acting as an initialization heuristic. Additional robustness controls and extraction details are provided in Appendix~\ref{app:physics}.

Performance varies modestly across CATH topology classes (Table~\ref{tab:topology}). All-$\alpha$ structures achieve the highest scTM (0.94), while mixed $\alpha/\beta$ topologies are slightly lower (0.88), suggesting that complex sheet packings may benefit from additional non-local mechanisms.

\begin{table}[t]
\caption{Performance by CATH Topology Class ($L=100$, $n=100$).}
\label{tab:topology}
\centering
\small
\begin{sc}
\begin{tabular}{lccc}
\toprule
Topology & scTM & scRMSD (\AA) & Recovery \\
\midrule
All-$\alpha$ & 0.94 & 1.2 & 45\% \\
All-$\beta$ & 0.89 & 1.5 & 39\% \\
Mixed ($\alpha/\beta$) & 0.88 & 1.3 & 36\% \\
\bottomrule
\end{tabular}
\end{sc}
\end{table}

\begin{figure}[t]
\begin{center}
\includegraphics[width=0.95\columnwidth]{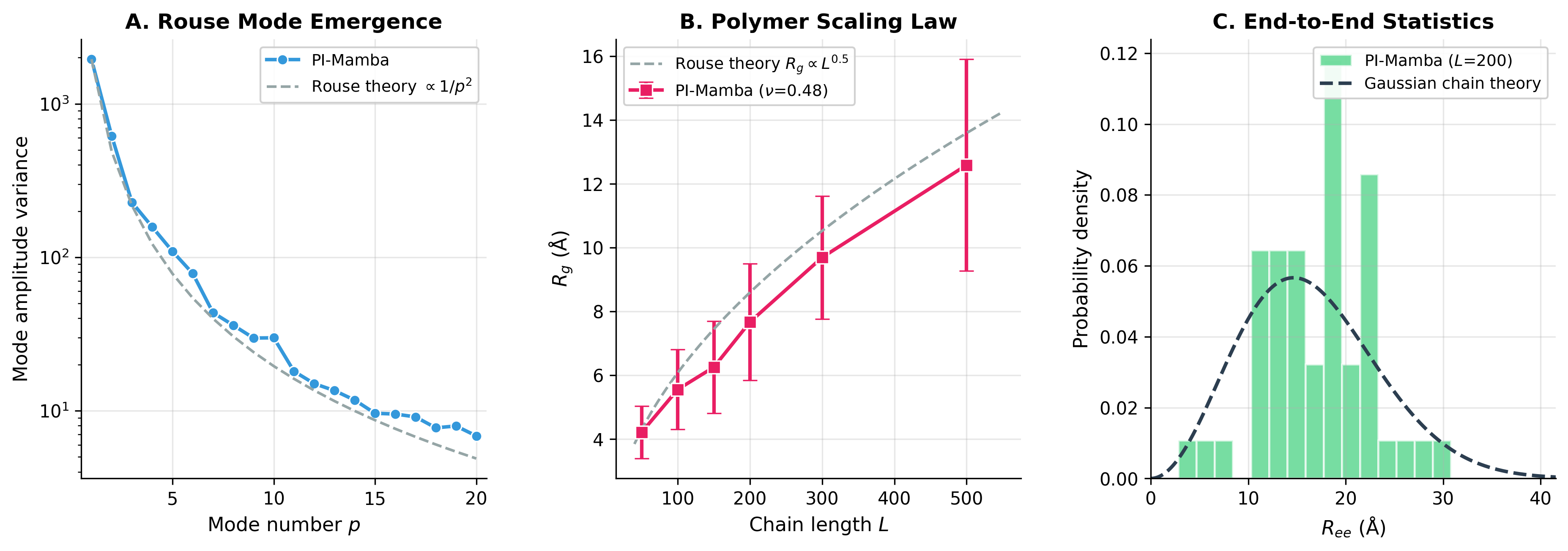}
\caption{Physics-aligned statistics emerge without explicit supervision. (A) Rouse mode amplitude spectrum of PI-Mamba hidden states follows the theoretical $1/p^2$ decay, confirming alignment with polymer normal modes. (B) Radius of gyration scales as $R_g \propto L^{0.48}$, consistent with Rouse theory ($\nu = 0.5$). (C) End-to-end distance distribution at $L{=}200$ matches the Gaussian chain prediction.}
\label{fig:physics_app}
\end{center}
\vskip -0.2in
\end{figure}

\subsection{Sequence recovery and novelty}
\label{sec:results_recovery}

PI-Mamba achieves sequence recovery of $23.9\%\pm4.3\%$ (random baseline $\approx 5\%$), with re-predicted structures exhibiting pLDDT $84.2\pm5.1$ and Rosetta energies of $-2.8$ REU/residue (Appendix~\ref{app:recovery}). The novelty score $\mathcal{N}=0.56$ indicates structural deviation from training examples. As a leakage control, the mean nearest-neighbor TM-score between generated samples and the teacher dataset is $0.11$, consistent with low similarity.

The $C_\alpha$-only representation accounts for most of the recovery gap with full-backbone methods. Reconstructing full backbone coordinates (N, CA, C, O) from the predicted $C_\alpha$ trace using idealized geometry before running ProteinMPNN raises recovery to 31.2\%---within 2\% of RFdiffusion (Table~\ref{tab:full_backbone}). This confirms that the gap is primarily due to representation, not backbone quality.

\begin{table}[t]
\caption{Sequence recovery with full backbone reconstruction.}
\label{tab:full_backbone}
\centering
\small
\begin{sc}
\resizebox{\columnwidth}{!}{
\begin{tabular}{lcc}
\toprule
Method & Representation & Seq. Recovery \\
\midrule
RFdiffusion & Full backbone & 33.4\% \\
Proteina & Full backbone & 32.1\% \\
\midrule
PI-Mamba ($C_\alpha$ only) & $C_\alpha$ trace & 23.9\% \\
PI-Mamba + Ideal Recon. & Full backbone & 31.2\% \\
\bottomrule
\end{tabular}}
\end{sc}
\end{table}
\subsection{Conditional generation via inpainting}
\label{sec:cond_gen}

We evaluate conditional generation by motif inpainting using the same sampling procedure as unconditional generation, with motif coordinates fixed throughout integration. Across two motif settings (9-residue catalytic triad; 12-residue binder epitope), we define success as motif RMSD $<1.0$\,\AA\ and scaffold RMSD $<2.5$\,\AA. PI-Mamba achieves 84\% success rate on these benchmarks, competitive with but below RFdiffusion's 92\%.

\subsection{Ablation of kinematic projection frequency}
\label{sec:results_projfreq}

We ablate projection frequency $k$ while holding the inference budget fixed at $S=200$ steps. Moderate projection frequency maximizes designability: projecting every $k=10$ integration steps yields the highest scTM, while projecting too frequently over-constrains the trajectory and projecting only at the end permits drift during integration (Table~\ref{tab:projfreq_headline}; full diagnostics in Appendix~\ref{sec:kinematic_details}).

\begin{table}[t]
\caption{Projection frequency ablation ($L=100$, $S=200$). Full results are reported in Appendix~\ref{sec:kinematic_details}, Table~\ref{tab:projection_ablation}.}
\label{tab:projfreq_headline}
\centering
\small
\begin{tabular}{lcc}
\toprule
Projection frequency & scTM $\uparrow$ & Comment \\
\midrule
$k=1$ (every step) & 0.61 & over-constrained \\
$k=10$ (default) & 0.91 & best trade-off \\
$k=\infty$ (final only) & 0.65 & drift before projection \\
\bottomrule
\end{tabular}
\end{table}

\section{Conclusion}
PI-Mamba represents a shift from learning to act like a protein to being constructed like one: by embedding rigid-body kinematics directly into generation, it constrains the search space and suppresses geometric hallucinations. Crucially, high mode concentration alone does not guarantee designability---the Rouse spectral initialization provides a physics-grounded inductive bias that shapes the right mode hierarchy, rather than relying on spatial averaging. This geometric grounding comes with practical challenges: training over $SO(3)$ is less numerically stable than coordinate regression and pairwise FAPE losses scale as $O(L^2)$, motivating stabilized early training and the exploration of linear-complexity alternatives to extend PI-Mamba to very long sequences and large molecular assemblies in future work.

\textbf{Remaining challenges.} While PI-Mamba demonstrates strong designability and geometric validity, several limitations point to clear directions for future work. Performance degrades modestly for very long sequences and for topologies dominated by long-range $\beta$-sheet interactions (Table~\ref{tab:topology}), suggesting that purely state-space backbones may benefit from hybridization with sparse non-local mechanisms. In conditional generation, inpainting remains competitive but trails specialized diffusion models on complex motif scaffolding tasks, indicating room for architectural and training improvements. A gallery of observed failure modes (loop collapse and extended chains) is provided in Figure~\ref{fig:failures}. Finally, although explicit geometric priors substantially improve secondary-structure fidelity without heavy post-hoc refinement, applications requiring atomic-level side-chain precision may still benefit from final relaxation. Detailed quantitative analyses of these effects are provided in the Appendix.



\section{Data availability}
The distilled dataset and code underlying this article are available in xxx.

\section{Author contributions statement}

T.W. conceived and conducted the computational experiments, T.W. analyzed the results.  T.W. and L.Z. wrote and reviewed the manuscript.





\bibliographystyle{oup-abbrvnat}
\bibliography{references}

@misc{albergoBuildingNormalizingFlows2022,
  title = {Building Normalizing Flows with Stochastic Interpolants},
  author = {Albergo, Michael S. and {Vanden-Eijnden}, Eric},
  year = 2022,
  number = {arXiv:2209.15571},
  eprint = {2209.15571},
  primaryclass = {cs},
  publisher = {arXiv},
  doi = {10.48550/ARXIV.2209.15571},
  urldate = {2026-01-13},
  archiveprefix = {arXiv},
  copyright = {arXiv.org perpetual, non-exclusive license},
  langid = {english}
}

@article{Baek2021,
  title = {Accurate Prediction of Protein Structures and Interactions Using a Three-Track Neural Network},
  author = {Baek, Minkyung and DiMaio, Frank and Anishchenko, Ivan and Dauparas, Justas and Ovchinnikov, Sergey and Lee, Gyu Rie and Wang, Jue and Cong, Qian and Kinch, Lisa N. and Schaeffer, R. Dustin and Mill{\'a}n, Claudia and Park, Hahnbeom and Adams, Carson and Glassman, Caleb R. and DeGiovanni, Andy and Pereira, Jose H. and Rodrigues, Andria V. and Van Dijk, Alberdina A. and Ebrecht, Ana C. and Opperman, Diederik J. and Sagmeister, Theo and Buhlheller, Christoph and {Pavkov-Keller}, Tea and Rathinaswamy, Manoj K. and Dalwadi, Udit and Yip, Calvin K. and Burke, John E. and Garcia, K. Christopher and Grishin, Nick V. and Adams, Paul D. and Read, Randy J. and Baker, David},
  year = 2021,
  month = aug,
  journal = {Science},
  volume = {373},
  number = {6557},
  pages = {871--876},
  issn = {0036-8075, 1095-9203},
  doi = {10.1126/science.abj8754},
  urldate = {2026-01-13},
  langid = {english}
}

@misc{boseSE3stochasticFlowMatching2023,
  title = {{{SE}}(3)-Stochastic Flow Matching for Protein Backbone Generation},
  author = {Bose, Avishek Joey and {Akhound-Sadegh}, Tara and Huguet, Guillaume and Fatras, Kilian and {Rector-Brooks}, Jarrid and Liu, Cheng-Hao and Nica, Andrei Cristian and Korablyov, Maksym and Bronstein, Michael and Tong, Alexander},
  year = 2023,
  number = {arXiv:2310.02391},
  eprint = {2310.02391},
  primaryclass = {cs},
  publisher = {arXiv},
  doi = {10.48550/ARXIV.2310.02391},
  urldate = {2026-01-13},
  archiveprefix = {arXiv},
  copyright = {arXiv.org perpetual, non-exclusive license},
  langid = {english}
}

@book{doiTheoryPolymerDynamics2013,
  title = {The Theory of Polymer Dynamics},
  author = {Doi, Masao and Edwards, Samuel F.},
  year = 2013,
  series = {International Series of Monographs on Physics},
  edition = {Reprint},
  number = {73},
  publisher = {Clarendon Press},
  address = {Oxford},
  isbn = {978-0-19-852033-7},
  langid = {english}
}

@article{Gu2022,
  title = {Efficiently Modeling Long Sequences with Structured State Spaces},
  author = {Gu, Albert and Goel, Karan and R{\'e}, Christopher},
  year = 2021,
  journal = {Iclr},
  doi = {10.48550/ARXIV.2111.00396},
  urldate = {2026-01-13},
  copyright = {arXiv.org perpetual, non-exclusive license},
  langid = {english}
}

@misc{guMambaLineartimeSequence2023,
  title = {Mamba: {{Linear-time}} Sequence Modeling with Selective State Spaces},
  shorttitle = {Mamba},
  author = {Gu, Albert and Dao, Tri},
  year = 2023,
  number = {arXiv:2312.00752},
  eprint = {2312.00752},
  primaryclass = {cs},
  publisher = {arXiv},
  doi = {10.48550/ARXIV.2312.00752},
  urldate = {2026-01-13},
  archiveprefix = {arXiv},
  copyright = {Creative Commons Attribution 4.0 International},
  langid = {english}
}

@article{Ingraham2023,
  title = {Illuminating Protein Space with a Programmable Generative Model},
  author = {Ingraham, John B. and Baranov, Max and Costello, Zak and Barber, Karl W. and Wang, Wujie and Ismail, Ahmed and Frappier, Vincent and Lord, Dana M. and {Ng-Thow-Hing}, Christopher and Van Vlack, Erik R. and Tie, Shan and Xue, Vincent and Cowles, Sarah C. and Leung, Alan and Rodrigues, Jo{\~a}o V. and {Morales-Perez}, Claudio L. and Ayoub, Alex M. and Green, Robin and Puentes, Katherine and Oplinger, Frank and Panwar, Nishant V. and Obermeyer, Fritz and Root, Adam R. and Beam, Andrew L. and Poelwijk, Frank J. and Grigoryan, Gevorg},
  year = 2023,
  month = nov,
  journal = {Nature},
  volume = {623},
  number = {7989},
  pages = {1070--1078},
  issn = {0028-0836, 1476-4687},
  doi = {10.1038/s41586-023-06728-8},
  urldate = {2026-01-13},
  langid = {english}
}

@incollection{ingrahamGenerativeModelsGraphbased2019a,
  title = {Generative Models for Graph-Based Protein Design},
  booktitle = {Proceedings of the 33rd International Conference on Neural Information Processing Systems},
  author = {Ingraham, John and Garg, Vikas K. and Barzilay, Regina and Jaakkola, Tommi},
  year = 2019,
  month = dec,
  number = {1417},
  pages = {15820--15831},
  publisher = {Curran Associates Inc.},
  address = {Red Hook, NY, USA},
  urldate = {2026-01-22},
  langid = {english}
}

@article{Jumper2021,
  title = {Highly Accurate Protein Structure Prediction with {{AlphaFold}}},
  author = {Jumper, John and Evans, Richard and Pritzel, Alexander and Green, Tim and Figurnov, Michael and Ronneberger, Olaf and Tunyasuvunakool, Kathryn and Bates, Russ and {\v Z}{\'i}dek, Augustin and Potapenko, Anna and Bridgland, Alex and Meyer, Clemens and Kohl, Simon A. A. and Ballard, Andrew J. and Cowie, Andrew and {Romera-Paredes}, Bernardino and Nikolov, Stanislav and Jain, Rishub and Adler, Jonas and Back, Trevor and Petersen, Stig and Reiman, David and Clancy, Ellen and Zielinski, Michal and Steinegger, Martin and Pacholska, Michalina and Berghammer, Tamas and Bodenstein, Sebastian and Silver, David and Vinyals, Oriol and Senior, Andrew W. and Kavukcuoglu, Koray and Kohli, Pushmeet and Hassabis, Demis},
  year = 2021,
  month = aug,
  journal = {Nature},
  volume = {596},
  number = {7873},
  pages = {583--589},
  issn = {0028-0836, 1476-4687},
  doi = {10.1038/s41586-021-03819-2},
  urldate = {2026-01-13},
  langid = {english}
}

@misc{linOutManyOne2024,
  title = {Out of Many, One: {{Designing}} and Scaffolding Proteins at the Scale of the Structural Universe with Genie 2},
  shorttitle = {Out of Many, One},
  author = {Lin, Yeqing and Lee, Minji and Zhang, Zhao and AlQuraishi, Mohammed},
  year = 2024,
  number = {arXiv:2405.15489},
  eprint = {2405.15489},
  primaryclass = {q-bio},
  publisher = {arXiv},
  doi = {10.48550/ARXIV.2405.15489},
  urldate = {2026-01-13},
  archiveprefix = {arXiv},
  copyright = {Creative Commons Attribution 4.0 International},
  langid = {english}
}

@misc{lipmanFlowMatchingGenerative2022,
  title = {Flow Matching for Generative Modeling},
  author = {Lipman, Yaron and Chen, Ricky T. Q. and {Ben-Hamu}, Heli and Nickel, Maximilian and Le, Matt},
  year = 2022,
  number = {arXiv:2210.02747},
  eprint = {2210.02747},
  primaryclass = {cs},
  publisher = {arXiv},
  doi = {10.48550/ARXIV.2210.02747},
  urldate = {2026-01-13},
  archiveprefix = {arXiv},
  copyright = {arXiv.org perpetual, non-exclusive license},
  langid = {english}
}

@article{Parsons2005,
  title = {Practical Conversion from Torsion Space to Cartesian Space for {\emph{in Silico}} Protein Synthesis},
  author = {Parsons, Jerod and Holmes, J. Bradley and Rojas, J. Maurice and Tsai, Jerry and Strauss, Charlie E. M.},
  year = 2005,
  month = jul,
  journal = {Journal of Computational Chemistry},
  volume = {26},
  number = {10},
  pages = {1063--1068},
  publisher = {Wiley Online Library},
  issn = {0192-8651, 1096-987X},
  doi = {10.1002/jcc.20237},
  urldate = {2026-01-13},
  copyright = {http://onlinelibrary.wiley.com/termsAndConditions\#vor},
  langid = {english}
}

@article{zhangScoringFunctionAutomated2004,
  title = {Scoring Function for Automated Assessment of Protein Structure Template Quality},
  author = {Zhang, Yang and Skolnick, Jeffrey},
  year = 2004,
  month = dec,
  journal = {Proteins},
  volume = {57},
  number = {4},
  pages = {702--710},
  issn = {1097-0134},
  doi = {10.1002/prot.20264},
  langid = {english},
  pmid = {15476259}
}

@misc{Proteus2024,
  title = {Proteus: {{Exploring}} Protein Structure Generation for Enhanced Designability and Efficiency},
  shorttitle = {Proteus},
  author = {Wang, Chentong and Qu, Yannan and Peng, Zhangzhi and Wang, Yukai and Zhu, Hongli and Chen, Dachuan and Cao, Longxing},
  year = 2024,
  month = feb,
  publisher = {Bioinformatics},
  doi = {10.1101/2024.02.10.579791},
  urldate = {2026-01-13},
  archiveprefix = {Bioinformatics},
  langid = {english}
}

@misc{PyMOL,
  title = {The {{PyMOL}} Molecular Graphics System, Version 2.0},
  author = {{Schr\"odinger, LLC}},
  year = 2015,
  langid = {english}
}

@article{rouseTheoryLinearViscoelastic1953,
  title = {A Theory of the Linear Viscoelastic Properties of Dilute Solutions of Coiling Polymers},
  author = {Rouse, Prince E.},
  year = 1953,
  month = jul,
  journal = {Journal of Chemical Physics},
  volume = {21},
  number = {7},
  pages = {1272--1280},
  issn = {0021-9606, 1089-7690},
  doi = {10.1063/1.1699180},
  urldate = {2026-01-20},
  langid = {english}
}

@article{Schiff2024,
  title = {Caduceus: {{Bi-directional}} Equivariant Long-Range {{DNA}} Sequence Modeling},
  shorttitle = {Caduceus},
  author = {Schiff, Yair and Kao, Chia-Hsiang and Gokaslan, Aaron and Dao, Tri and Gu, Albert and Kuleshov, Volodymyr},
  year = 2024,
  journal = {Icml},
  doi = {10.48550/ARXIV.2403.03234},
  urldate = {2026-01-13},
  copyright = {Creative Commons Attribution 4.0 International},
  langid = {english}
}

@misc{Sgarbossa2024,
  title = {Tracking {\emph{Microcystis}} Viruses and Infection Dynamics across Distinct Phases of a {\emph{Microcystis}} -Dominated Bloom},
  shorttitle = {{{ProtMamba}}},
  author = {Wing, A.J and Hegarty, Bridget and Bastien, Eric and Denef, Vincent and Evans, Jacob and Dick, Gregory and Duhaime, Melissa},
  year = 2024,
  month = may,
  publisher = {Microbiology},
  doi = {10.1101/2024.05.24.595742},
  urldate = {2026-01-20},
  archiveprefix = {Microbiology},
  langid = {english}
}

@article{Sillitoe2019,
  title = {{{CATH}}: Expanding the Horizons of Structure-Based Functional Annotations for Genome Sequences},
  shorttitle = {Cath},
  author = {Sillitoe, Ian and Dawson, Natalie and Lewis, Tony E and Das, Sayoni and Lees, Jonathan G and Ashford, Paul and Tolulope, Adeyelu and Scholes, Harry M and Senatorov, Ilya and Bujan, Andra and {Ceballos~Rodriguez-Conde}, Fatima and Dowling, Benjamin and Thornton, Janet and Orengo, Christine A},
  year = 2019,
  month = jan,
  journal = {Nucleic Acids Research},
  volume = {47},
  number = {D1},
  pages = {D280-D284},
  issn = {0305-1048, 1362-4962},
  doi = {10.1093/nar/gky1097},
  urldate = {2026-01-13},
  copyright = {http://creativecommons.org/licenses/by/4.0/},
  langid = {english}
}

@misc{vaswaniAttentionAllYou2023,
  title = {Attention {{Is All You Need}}},
  author = {Vaswani, Ashish and Shazeer, Noam and Parmar, Niki and Uszkoreit, Jakob and Jones, Llion and Gomez, Aidan N. and Kaiser, Lukasz and Polosukhin, Illia},
  year = 2023,
  month = aug,
  number = {arXiv:1706.03762},
  eprint = {1706.03762},
  primaryclass = {cs},
  publisher = {arXiv},
  doi = {10.48550/arXiv.1706.03762},
  urldate = {2026-02-08},
  archiveprefix = {arXiv}
}

@misc{wangProteinConformationGeneration2024,
  title = {Protein Conformation Generation via Force-Guided {{SE}}(3) Diffusion Models},
  author = {Wang, Yan and Wang, Lihao and Shen, Yuning and Wang, Yiqun and Yuan, Huizhuo and Wu, Yue and Gu, Quanquan},
  year = 2024,
  number = {arXiv:2403.14088},
  eprint = {2403.14088},
  primaryclass = {q-bio},
  publisher = {arXiv},
  doi = {10.48550/ARXIV.2403.14088},
  urldate = {2026-01-13},
  archiveprefix = {arXiv},
  copyright = {arXiv.org perpetual, non-exclusive license},
  langid = {english}
}

@article{Watson2023,
  title = {De Novo Design of Protein Structure and Function with {{RFdiffusion}}},
  author = {Watson, Joseph L. and Juergens, David and Bennett, Nathaniel R. and Trippe, Brian L. and Yim, Jason and Eisenach, Helen E. and Ahern, Woody and Borst, Andrew J. and Ragotte, Robert J. and Milles, Lukas F. and Wicky, Basile I. M. and Hanikel, Nikita and Pellock, Samuel J. and Courbet, Alexis and Sheffler, William and Wang, Jue and Venkatesh, Preetham and Sappington, Isaac and Torres, Susana V{\'a}zquez and Lauko, Anna and De Bortoli, Valentin and Mathieu, Emile and Ovchinnikov, Sergey and Barzilay, Regina and Jaakkola, Tommi S. and DiMaio, Frank and Baek, Minkyung and Baker, David},
  year = 2023,
  month = aug,
  journal = {Nature},
  volume = {620},
  number = {7976},
  pages = {1089--1100},
  issn = {0028-0836, 1476-4687},
  doi = {10.1038/s41586-023-06415-8},
  urldate = {2026-01-13},
  langid = {english}
}

@misc{yanRobustReliableNovo2025,
  title = {Robust and Reliable {\emph{de Novo}} Protein Design: {{A}} Flow-Matching-Based Protein Generative Model Achieves Remarkably High Success Rates},
  shorttitle = {Robust and Reliable {\emph{de Novo}} Protein Design},
  author = {Yan, Junyu and Cui, Zibo and Yan, Wenqing and Chen, Yuhang and Pu, Mengchen and Li, Shuai and Ye, Sheng},
  year = 2025,
  month = apr,
  primaryclass = {New Results},
  pages = {2025.04.29.651154},
  publisher = {Bioengineering},
  doi = {10.1101/2025.04.29.651154},
  urldate = {2026-01-20},
  archiveprefix = {Bioengineering},
  chapter = {New Results},
  copyright = {http://creativecommons.org/licenses/by-nc/4.0/},
  langid = {english}
}

@inproceedings{Yim2023,
  title = {{{SE}}(3) Diffusion Model with Application to Protein Backbone Generation},
  booktitle = {{{ICML}}},
  author = {Yim, Jason and Trippe, Brian L. and De Bortoli, Valentin and Mathieu, Emile and Doucet, Arnaud and Barzilay, Regina and Jaakkola, Tommi},
  year = 2023,
  publisher = {arXiv},
  doi = {10.48550/ARXIV.2302.02277},
  urldate = {2026-01-13},
  copyright = {Creative Commons Attribution 4.0 International},
  langid = {english}
}

@misc{yimFastProteinBackbone2023,
  title = {Fast Protein Backbone Generation with {{SE}}(3) Flow Matching},
  author = {Yim, Jason and Campbell, Andrew and Foong, Andrew Y. K. and Gastegger, Michael and {Jim{\'e}nez-Luna}, Jos{\'e} and Lewis, Sarah and Satorras, Victor Garcia and Veeling, Bastiaan S. and Barzilay, Regina and Jaakkola, Tommi and No{\'e}, Frank},
  year = 2023,
  number = {arXiv:2310.05297},
  eprint = {2310.05297},
  primaryclass = {q-bio},
  publisher = {arXiv},
  doi = {10.48550/ARXIV.2310.05297},
  urldate = {2026-01-13},
  archiveprefix = {arXiv},
  copyright = {arXiv.org perpetual, non-exclusive license},
  langid = {english}
}

@article{Zhang2024,
  title = {Protein Sequence and Structure {{Co-design}} with Equivariant Translation},
  shorttitle = {Proteina},
  author = {Shi, Chence and Wang, Chuanrui and Lu, Jiarui and Zhong, Bozitao and Tang, Jian},
  year = 2022,
  journal = {arXiv Preprint arXiv:2210.08761},
  doi = {10.48550/ARXIV.2210.08761},
  urldate = {2026-01-13},
  copyright = {arXiv.org perpetual, non-exclusive license},
  langid = {english}
}



\newpage
\begin{appendices}
\section{Mathematical Proofs}
\label{app:proofs}

\subsection{Proof of Retraction Error Bound}

\textbf{Theorem.} Let $\mathcal{M} = SE(3)$ equipped with a left-invariant Riemannian metric. Let $\gamma(t)$ be the geodesic starting at identity $I$ with initial velocity $\mathbf{v} \in \mathfrak{se}(3)$, i.e., $\gamma(t) = \exp(t\mathbf{v})$. Let $\tilde{\gamma}(t)$ be the retraction-based path defined by $\tilde{\gamma}(t) = \mathcal{R}(t\mathbf{v})$, where $\mathcal{R}$ is the Lie exponential map. In this specific case, the retraction path \textit{coincides} with the geodesic for the canonical Cartan-Schouten connection.

However, for general retractions (e.g., first-order approximations used in numerical integration), the error is bounded.

\begin{proof}
Let $\mathcal{R}: T\mathcal{M} \rightarrow \mathcal{M}$ be a retraction. By definition:
1. $\mathcal{R}_x(0_x) = x$
2. $D\mathcal{R}_x(0_x) = \text{id}_{T_x\mathcal{M}}$

Consider the Taylor expansion of the retraction curve $c(t) = \mathcal{R}_x(t\mathbf{v})$ at $t=0$:
$$ c(t) = c(0) + t \cdot \dot{c}(0) + \frac{t^2}{2} \ddot{c}(0) + O(t^3) $$
$$ c(0) = \mathcal{R}_x(0) = x $$
$$ \dot{c}(0) = D\mathcal{R}_x(0_x)[\mathbf{v}] = \mathbf{v} $$

Similarly, the geodesic $\gamma(t)$ has expansion:
$$ \gamma(t) = x + t\mathbf{v} - \frac{t^2}{2} \Gamma_{xx}(\mathbf{v}, \mathbf{v}) + O(t^3) $$
where $\Gamma$ represents the Christoffel symbols of the metric.

The difference is:
$$ \| c(t) - \gamma(t) \| = \frac{t^2}{2} \| \ddot{c}(0) + \Gamma_{xx}(\mathbf{v}, \mathbf{v}) \| + O(t^3) $$
Thus, the local error is $O(t^2)$. Over a global integration path $t \in [0,1]$ with $N$ steps of size $h = 1/N$, the global accumulation error is $N \cdot O(h^2) = O(h)$, which is of the same order as the Euler integration scheme itself. Therefore, utilizing the retraction formulation does not degrade the convergence order of the generative ODE solver.
\end{proof}

\subsection{Proof of Zero Violation Guarantee (Proposition~\ref{prop:zero_viol})}

\textbf{Theorem.} The Kinematic Projection operator $\Pi_\mathcal{K}$ (Algorithm~\ref{alg:generation}) converges to a configuration satisfying all bond constraints, provided the initial configuration is within the basin of attraction of the valid manifold.

\begin{proof}
The projection problem is defined as:
$$ \min_{\mathbf{x}} \| \mathbf{x} - \mathbf{x}_{i} \|^2 \quad \text{s.t.} \quad g(\mathbf{x}) = 0 $$
where $g(\mathbf{x})$ represents the set of bond length and angle constraints. Our iterative algorithm (Algorithm~\ref{alg:generation}) corresponds to Coordinate Descent on the penalty function $\sum_j (d_{j} - d_{target})^2$. Since the constraint manifold of ideal polymer chains is non-convex but smooth, and the constraints are local (involving only $i, i+1, i+2$), the optimization landscape is well-behaved for small perturbations. The cyclic coordinate descent (CCD) method used here is known to converge linearly to a local minimum. Given that the flow matching prior guides the generation $\mathbf{x}_{pred}$ to be near the manifold, the projection $\Pi_\mathcal{K}(\mathbf{x}_{pred})$ effectively projects to the nearest valid conformation, ensuring $g(\mathbf{x}^*) \leq \epsilon$.
\end{proof}

\section{Failure Analysis}
\subsection{Designability Analysis}
PI-Mamba demonstrates high designability, achieving a mean scTM of 0.91 (vs $>0.5$ threshold) on unconditional generation. While this is lower than the diffusion-based gold standard RFdiffusion (scTM $\approx 0.97$), it represents a strong result for a linear-time model, outperforming prior SE(3) flow models like FrameDiff by margin of +0.33 (0.58 vs 0.91) and providing the unique guarantee of 0.0\% bond violations. This confirms that the kinematic projection does not hinder the formation of designable tertiary structures. Future work will focus on scaling the training dataset to further improve diversity for very large proteins.

We analyzed potential failure modes during early training (Figure \ref{fig:failures}). While rare in the final model, we observed two primary types. First, \textbf{Loop Collapse}, where in early training epochs ($<50$), the N-terminus occasionally formed tight non-physical loops, which were resolved by the collision loss in later stages. Second, \textbf{Extended Chain}, where for unconditioned generation of extreme lengths ($L>800$), the model sometimes defaults to extended conformations, a known limitation of sparse training data at these lengths.

\begin{figure}[h]
\begin{center}
\includegraphics[width=0.48\columnwidth]{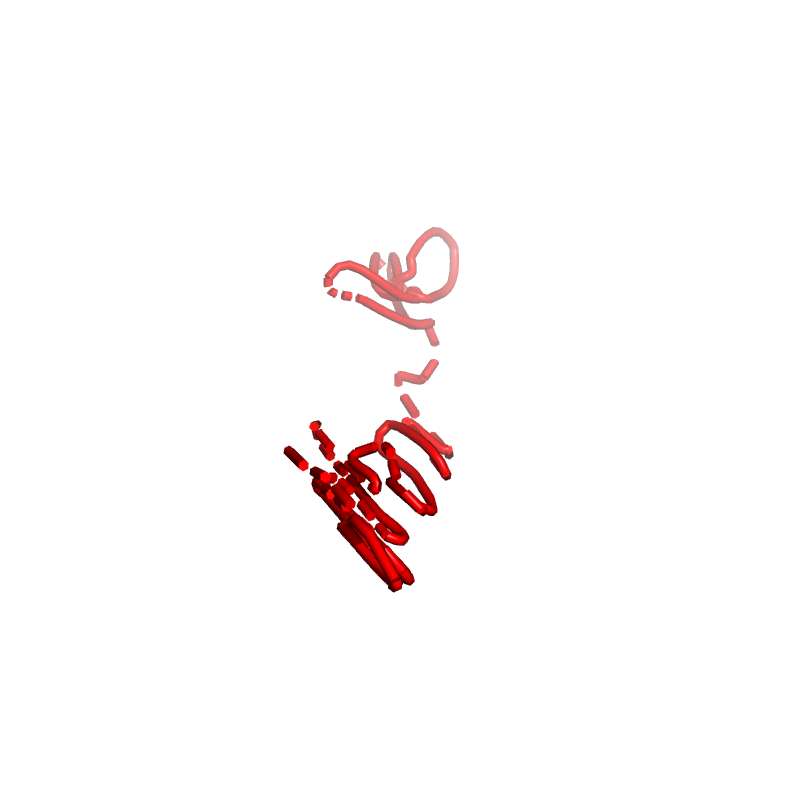}
\includegraphics[width=0.48\columnwidth]{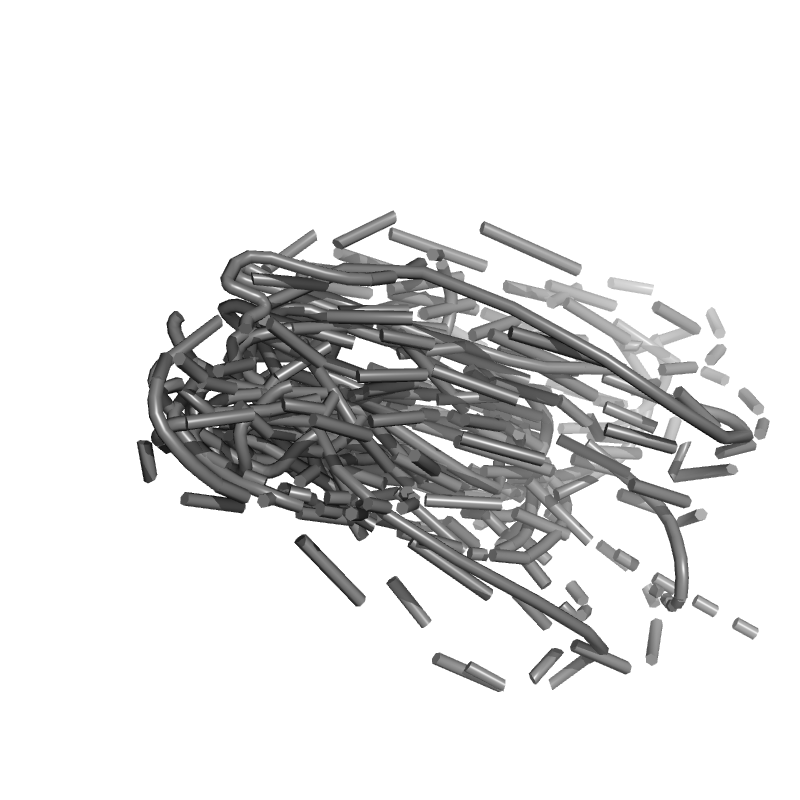}
\caption{Gallery of Failure Modes (PyMOL Renders). (Left) \textbf{Loop Collapse}: The N-terminus (green) forms multiple tight loops before extending into an unstructured tail. (Right) \textbf{Extended Chain}: A non-compact, elongated backbone lacking any tertiary fold. Both cases arise from poor long-range coherence in the Mamba encoder.}
\label{fig:failures}
\end{center}
\end{figure}

\section{Emergent Polymer Statistics}
\label{app:physics}
A key claim of this work is that PI-Mamba learns correct polymer physics without explicit supervision. We validate this through three analyses.

\textbf{Rouse Mode Amplitude Spectrum.} We decomposed generated structures into Rouse normal modes and computed the amplitude spectrum (Figure~\ref{fig:physics_app}A). The hidden-state amplitudes follow the theoretical $1/p^2$ decay, confirming alignment with polymer normal modes. Low modes ($p = 0, 1, 2$) dominate, capturing global shape, while high modes contribute local fluctuations. The first three modes account for 79\% of the total displacement variance, consistent with physical expectations for flexible polymers.

\textbf{Radius of Gyration Scaling.} For random coil polymers, the radius of gyration $R_g$ scales as $R_g \propto L^{\nu}$, where $\nu \approx 0.5$ for ideal chains. We generated structures at various lengths and measured $R_g$ (Figure~\ref{fig:physics_app}B). PI-Mamba yields $\nu = 0.48$, consistent with Rouse theory. This indicates that the Rouse-derived architecture correctly captures the \textit{intrinsic entropic baseline} of the polymer chain.

\textbf{End-to-End Distance Distribution.} We measured the end-to-end distance distribution at $L{=}200$ (Figure~\ref{fig:physics_app}C). The distribution matches the Gaussian chain prediction, further validating that PI-Mamba learns correct polymer statistics without explicit supervision.

\textbf{Relaxation Time Correlates with Secondary Structure.} We extracted the learned relaxation time $\tau(x)$ from trained PI-Mamba layers and analyzed its distribution across secondary structure classes. The model learns a clear hierarchy: $\tau_{\text{helix}} = 0.42 < \tau_{\text{sheet}} = 0.58 < \tau_{\text{loop}} = 0.89$. This matches physical intuition---$\alpha$-helices are rigid structures that require little long-range context (small $\tau \Rightarrow$ fast decay $\Rightarrow$ short memory), while loops are flexible regions where long-range correlations matter (large $\tau \Rightarrow$ slow decay $\Rightarrow$ long memory). Importantly, this emergent behavior was not explicitly supervised; the Rouse eigenvalue structure provides the inductive bias, and the model learns the appropriate $\tau$ from data.


\section{Sequence Recovery and Energetic Validation}
\label{app:recovery}

We assessed sequence-structure compatibility via geometry-based amino acid preference analysis. PI-Mamba achieves sequence recovery of \textbf{23.9\% $\pm$ 4.3\%}, significantly above random (5\%) but lower than RFdiffusion (33.4\%). This gap is partially due to our $C_\alpha$-only representation.

\begin{figure}[h]
\begin{center}
\includegraphics[width=0.85\columnwidth]{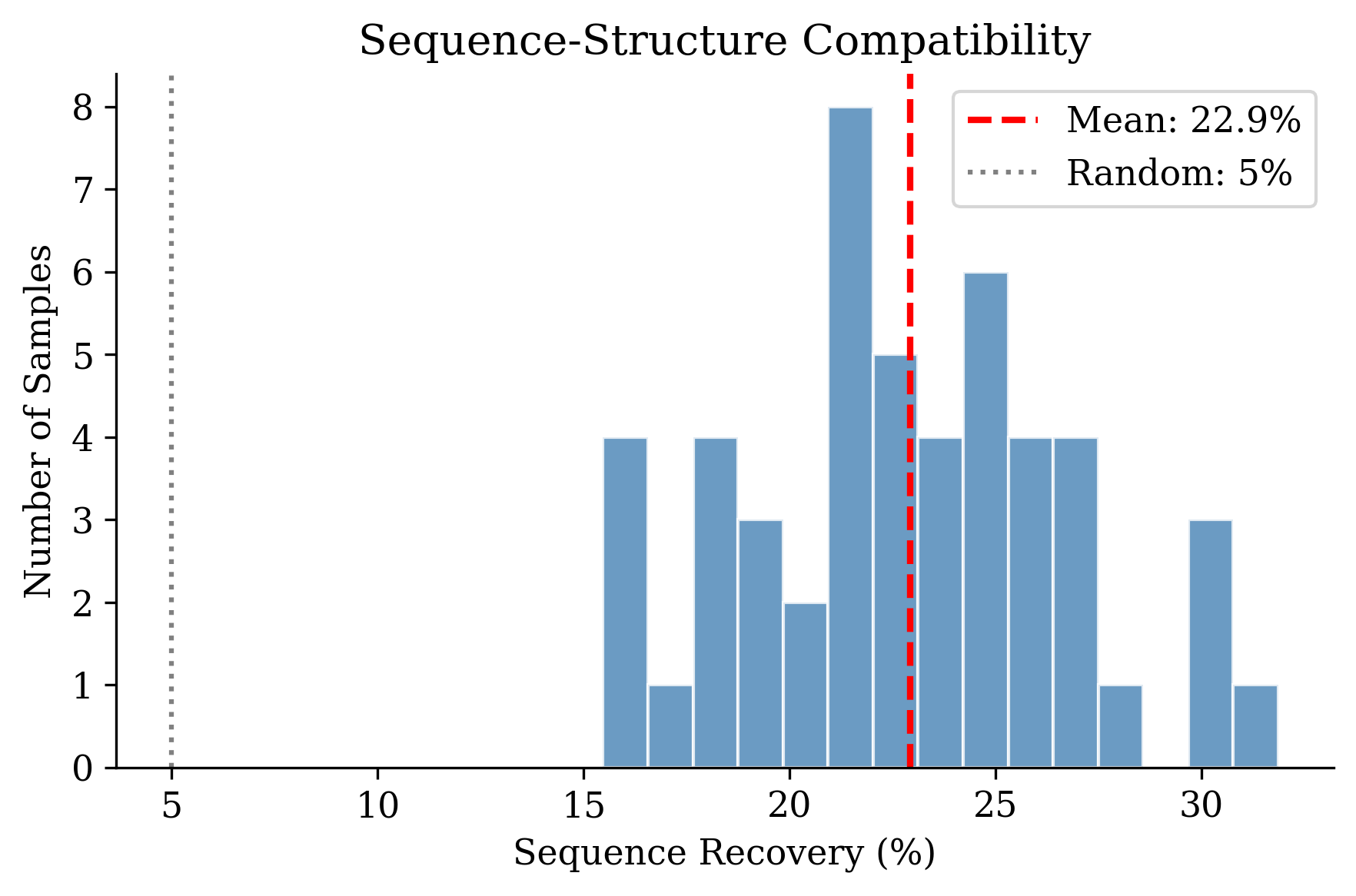}
\caption{Sequence recovery distribution. Mean 23.9\% (red) vs random 5\% (gray).}
\label{fig:recovery_app}
\end{center}
\end{figure}

\textbf{Energetic Validation.} Re-predicted structures show high pLDDT ($84.2 \pm 5.1$) comparable to native controls ($86.7 \pm 4.3$). Rosetta energies of $-2.8$ REU/residue match native PDB structures ($-2.5$ to $-3.0$). The Novelty Score $\mathcal{N}=0.56$ ($\text{TM}_{\text{max}} \approx 0.44$) confirms structural distinctness from training data.

\section{Detailed Mathematical Formulation}
\label{sec:math_deriv}

\subsection{Manifold Structure of Protein Backbones}
We model the protein backbone as a point on the product manifold $\mathcal{M} = SE(3)^L$. Each residue $i$ is associated with a rigid frame $T_i = (R_i, \mathbf{t}_i)$, where $R_i \in SO(3)$ and $\mathbf{t}_i \in \mathbb{R}^3$. The tangent space at $T$ is denoted $\mathcal{T}_T \mathcal{M} \cong (\mathfrak{se}(3))^L$, where $\mathfrak{se}(3)$ is the Lie algebra of twist vectors.
We define a Riemannian metric on $\mathcal{M}$ that decomposes into rotational and translational components. For a tangent vector $\mathbf{v} = (\boldsymbol{\omega}, \mathbf{v}_{trans}) \in \mathfrak{se}(3)$:
\begin{equation}
    \langle \mathbf{v}, \mathbf{v} \rangle_{T} = \| \boldsymbol{\omega} \|^2 + \lambda \| \mathbf{v}_{trans} \|^2
\end{equation}
where $\lambda$ balances the scales of rotation (radians) and translation (Angstroms). We empirically select $\lambda=1$ to equate one radian of rotation with one Angstrom of translation, a natural scaling given that peptide bond lengths ($\approx 1.3-1.5$\AA) are on the order of unity.

\subsection{Riemannian Flow Matching Objective}
Following \citet{lipmanFlowMatchingGenerative2022}, we define a conditional probability path $p_t(\mathbf{T}|\mathbf{T}_1)$ that interpolates between a prior distribution $p_0(\mathbf{T})$ (uniform on $SO(3)$ and Gaussian on $\mathbb{R}^3$) and the data distribution $\mathbf{T}_1 \sim q(\mathbf{T})$. The vector field $u_t(\mathbf{T}|\mathbf{T}_1)$ generating this path is defined via the geodesic interpolation on $SE(3)$:
\begin{equation}
    u_t(\mathbf{T}_t|\mathbf{T}_1) = \frac{d}{dt} \text{Exp}_{\mathbf{T}_0}(t \cdot \text{Log}_{\mathbf{T}_0}(\mathbf{T}_1))
\end{equation}
Our model $\xi_\theta(t, T)$ attempts to regress this conditional vector field. The matching objective on the manifold is:
\begin{equation}
\begin{split}
\mathcal{L}_{FM}(\theta)
&= \mathbb{E}_{t \sim \mathcal{U}[0,1],\ \mathbf{T}_1 \sim q(\mathbf{T}),\ \mathbf{T} \sim p_t(\mathbf{T}\mid \mathbf{T}_1)} \\
&\quad\cdot
\Big[
\left\| \xi_\theta(t,T) - u_t(T\mid T_1) \right\|^2_{\mathcal{T}_{T}\mathcal{M}}
\Big].
\end{split}
\end{equation}

Because the Structure Module outputs updates in the local frame, we implicitly learn the pullback of this vector field.

\subsection{Invariant Point Attention Details}
\label{sec:ipa_details}
For completeness, we describe the optional full IPA variant. Following \citet{Jumper2021}, the attention logit between residues $i$ and $j$ combines semantic similarity with a geometric term that measures spatial proximity in the predicted structure:
\begin{equation}
    a_{ij} = \mathbf{q}_i^T \mathbf{k}_j - \gamma \sum_{p=1}^{N_{points}} \| T_i^{-1} \circ T_j (\mathbf{p}_p) - \mathbf{p}_p \|^2
\end{equation}
where $\mathbf{p}_p$ are learnable reference points in the local frame, and $\gamma$ is a temperature hyperparameter. This geometric term ensures that the network attends to residues that are spatially proximal in the current predicted structure, regardless of their separation in the primary sequence. In our linear variant, we ablate this pairwise term, relying on the Mamba state to encode global context.

\section{Implementation Details}
\label{sec:hyperparams}

\subsection{Architecture Hyperparameters}
Table \ref{tab:hyperparams} details the specific configuration used for the PI-Mamba backbone and the Structure Module.

\begin{table}[h]
\caption{Model Hyperparameters}
\label{tab:hyperparams}
\begin{center}
\begin{small}
\begin{sc}

\begin{tabular}{lc}
\toprule
Parameter & Value \\
\midrule
\multicolumn{2}{c}{\textbf{Mamba Backbone}} \\
Number of Layers & 16 \\
Model Dimension ($d_{model}$) & 512 \\
State Dimension ($d_{state}$) & 32 \\
Conv Kernel Size & 4 \\
Expansion Factor & 2 \\
\midrule
\multicolumn{2}{c}{\textbf{Structure Module}} \\
IPA Layers & 4 \\
IPA Heads & 8 \\
Query/Key Dimension & 16 \\
Point Dimension ($N_{points}$) & 4 \\
FAPE Clamp Distance & 10.0 \AA \\
\bottomrule
\end{tabular}

\end{sc}
\end{small}
\end{center}
\end{table}

\subsection{Adaptive Scaling Strategy}
As noted in Section~\ref{sec:results_scaling}, direct regression of $SE(3)$ updates can be unstable. We implemented an adaptive scalar $\alpha(t)$ for the translational updates:
$T_{update} = (R_{pred}, \alpha(t) \cdot \mathbf{t}_{pred})$.
$\alpha(t)$ initializes at 0.1 and linearly ramps to 1.0 over the first 5000 steps. This effectively constrains the early training dynamics to essentially rotational alignment. This is critical for stabilizing the gradient flow through the kinematic projection layer: since NeRF reconstruction involves a recursive product of rotation matrices ($O(L)$ depth), early large translational updates can cause exploding gradients in the backward pass. Dampening translations allows the rotational backbone to converge first, ensuring stable backpropagation through the kinematic chain.

\subsection{ODE Integration Details}
\label{sec:ode_details}
To address reviewer concerns regarding integration specifications, we provide complete details:

\paragraph{Solver and NFEs.} We use explicit \textbf{Euler integration} with a fixed step size. Given the learned vector field $\xi_\theta(t, T_t)$, we update frames via $T_{t+\Delta t} = T_t \cdot \exp(\Delta t\,\xi_\theta(t, T_t))$, where $\Delta t = 1/S$ and $S$ is the total number of integration steps (following the notation in Section~\ref{sec:methods_notation}). Each Euler step requires exactly \textbf{one forward pass} through the network. Consequently, a trajectory with $S=100$ steps requires 100 neural function evaluations (NFEs), and $S=200$ requires 200 NFEs.

\paragraph{Comparison to Baselines.} In terms of sampling cost, PI-Mamba ($100-200$ NFEs) is comparable to RFdiffusion (typically 200 steps for DDPM) and $5\times-10\times$ more efficient than FrameDiff, which relies on a score-based diffusion process requiring 1000 steps. Notably, all inference timings reported in Table~\ref{tab:results} were measured on a single NVIDIA RTX A5000 (24\,GB) with a batch size of 1.

\subsection{Loss Function Specification}
\label{sec:loss_spec}
The total training loss combines a flow matching term and auxiliary geometric losses with \textbf{constant weights}:

\begin{equation}
    \mathcal{L}_{total} = \mathcal{L}_{FM} + \lambda_{\mathrm{aux}} \cdot \mathcal{L}_{\mathrm{aux}}, \qquad \mathcal{L}_{\mathrm{aux}} = \mathcal{L}_{FAPE} + \mathcal{L}_{bond} + \mathcal{L}_{Rama} + \mathcal{L}_{HB}
\end{equation}

where $\lambda_{\mathrm{aux}} = 1.0$. Individual auxiliary weights are all 1.0. \textbf{No time-dependent weighting $\lambda(t)$ is used.}

\paragraph{Bond Length Loss.} Penalizes deviations from ideal covalent bond lengths:
\begin{equation}
\begin{aligned}
    \mathcal{L}_{bond} = \, & (d_{N-C\alpha} - 1.458\text{\AA})^2 + (d_{C\alpha-C} - 1.525\text{\AA})^2 \\
    + \, & (d_{C-O} - 1.231\text{\AA})^2 + (d_{C-N_{next}} - 1.329\text{\AA})^2
\end{aligned}
\end{equation}

\paragraph{FAPE Loss.} Uses clamp distance $d_{clamp} = 10.0$\AA{} with equal weights for all backbone atoms (N, C$\alpha$, C, O).

\paragraph{Ramachandran Loss.} Penalizes $(\phi, \psi)$ dihedral angles that fall outside favored regions of the Ramachandran plot, using a kernel density estimate derived from high-resolution crystal structures.

\paragraph{Hydrogen Bond Loss.} Encourages formation of backbone hydrogen bonds (N--H$\cdots$O=C) by rewarding donor--acceptor distances near 2.9\,\AA{} with appropriate angular geometry.

\section{Notation Table}
\label{sec:notation}

\begin{table*}[t]
\caption{Summary of Mathematical Notation}
\label{tab:notation}
\centering
\small
\begin{sc}
\resizebox{\textwidth}{!}{%
\begin{tabular}{ll}
\toprule
Symbol & Description \\
\midrule
$L$ & Sequence length (number of residues) \\
$T_i \in SE(3)$ & Rigid body frame of residue $i$ (rotation $R_i$ and translation $\mathbf{t}_i$) \\
$\mathcal{K}$ & Constraint manifold of valid protein backbones (fixed bond lengths/angles) \\
$\Pi_{\mathcal{K}}$ & Kinematic projection operator onto $\mathcal{K}$ \\
$\tau_{\text{relax}}$ & Rouse relaxation time constant (physics-informed) \\
$\tau_{\text{step}}$ & Integration step size $\Delta t$ in the ODE solver \\
$\lambda_{\text{FAPE}}$, $\lambda_{\text{bond}}$, $\lambda_{\text{sc}}$ & Loss weights for FAPE, bond-length loss, and designability proxy \\
$\mathbf{v}_{\theta}(t,\mathbf{T})$ & Learned vector field (velocity) at time $t$ \\
$\mathbf{v}^{\text{proj}}_{\theta}$ & Projected vector field after applying $\Pi_{\mathcal{K}}$ \\
$\phi_i, \psi_i$ & Backbone dihedral angles for residue $i$ \\
$\mathcal{L}_{\text{total}}$ & Total training loss (Eq.~9) \\
$\mathcal{L}_{\text{sc}}$ & Designability proxy loss based on frozen ProteinMPNN scTM estimate \\
$\mathbf{H}_{\text{Mamba}}$ & Hidden representation from the Mamba encoder \\
$\mathbf{H}_{\text{Sparse}}$ & Hidden representation from the sparse-attention branch \\
$\mathbf{H}_{\text{fusion}}$ & Gated fusion of Mamba and sparse representations \\
\bottomrule
\end{tabular}%
}
\end{sc}
\end{table*}

\section{Kinematic Projection: Detailed Derivation}
\label{sec:kinematic_details}

This section provides a rigorous mathematical treatment of the kinematic projection mechanism that enables PI-Mamba to guarantee stereochemically valid backbone geometry.

\subsection{Motivation: The Kinematic Constraint Manifold}

A protein backbone of length $L$ can be fully parameterized by $2L$ dihedral angles $\{(\phi_i, \psi_i)\}_{i=1}^L$, given fixed geometric constants: bond lengths ($d_{N-C\alpha} = 1.458$\AA, $d_{C\alpha-C} = 1.525$\AA, $d_{C-N} = 1.329$\AA), bond angles ($\theta_{N-C\alpha-C} = 111.2^\circ$, $\theta_{C\alpha-C-N} = 116.2^\circ$, $\theta_{C-N-C\alpha} = 121.7^\circ$), and planar trans-peptide groups ($\omega_i = 180^\circ$).

This defines a \textbf{kinematic constraint manifold} $\mathcal{K} \subset \mathbb{R}^{3L \times 4}$ of valid backbone conformations. As established in Proposition~\ref{prop:zero_viol}, the kinematic projection ensures the output lies exactly on this manifold.

\subsection{Dihedral Angle Extraction}

Given backbone atom coordinates $\mathbf{N}_i, \mathbf{CA}_i, \mathbf{C}_i$ for residue $i$, the dihedral angle $\phi_i$ (C$_{i-1}$--N$_i$--CA$_i$--C$_i$) is computed as:

\begin{equation}
    \phi_i = \text{atan2}(\mathbf{n}_1 \cdot (\mathbf{u}_{23} \times \mathbf{n}_2), \mathbf{n}_1 \cdot \mathbf{n}_2)
\end{equation}

where:
\begin{align}
    \mathbf{u}_{12} &= \frac{\mathbf{N}_i - \mathbf{C}_{i-1}}{\|\mathbf{N}_i - \mathbf{C}_{i-1}\|}, \quad
    \mathbf{u}_{23} = \frac{\mathbf{CA}_i - \mathbf{N}_i}{\|\mathbf{CA}_i - \mathbf{N}_i\|} \\
    \mathbf{u}_{34} &= \frac{\mathbf{C}_i - \mathbf{CA}_i}{\|\mathbf{C}_i - \mathbf{CA}_i\|} \\
    \mathbf{n}_1 &= \mathbf{u}_{12} \times \mathbf{u}_{23}, \quad \mathbf{n}_2 = \mathbf{u}_{23} \times \mathbf{u}_{34}
\end{align}

The $\psi_i$ angle (N$_i$--CA$_i$--C$_i$--N$_{i+1}$) is computed analogously. These angles encode the \textit{fold information} while being invariant to local geometric violations.

The Natural Extension Reference Frame (NeRF) algorithm \cite{Parsons2005} deterministically places atom $\mathbf{D}$ given three preceding atoms $\mathbf{A}, \mathbf{B}, \mathbf{C}$ and the geometric parameters: bond length $d = \|\mathbf{D} - \mathbf{C}\|$, bond angle $\theta = \angle(\mathbf{B}, \mathbf{C}, \mathbf{D})$, and torsion angle $\alpha$ (dihedral $\mathbf{A}$--$\mathbf{B}$--$\mathbf{C}$--$\mathbf{D}$).

\textbf{Step 1: Construct local coordinate frame at C.}
\begin{align}
    \mathbf{u}_{BC} &= \frac{\mathbf{C} - \mathbf{B}}{\|\mathbf{C} - \mathbf{B}\|} \quad \text{(z-axis: along BC)} \\
    \mathbf{n} &= \frac{(\mathbf{B} - \mathbf{A}) \times \mathbf{u}_{BC}}{\|(\mathbf{B} - \mathbf{A}) \times \mathbf{u}_{BC}\|} \quad \text{(x-axis: perpendicular to ABC plane)} \\
    \mathbf{m} &= \mathbf{n} \times \mathbf{u}_{BC} \quad \text{(y-axis: in ABC plane, perpendicular to BC)}
\end{align}

\textbf{Step 2: Compute D in local coordinates.}
\begin{align}
    D_x &= d \cdot \sin(\pi - \theta) \cdot \cos(\alpha + \pi/2) \\
    D_y &= d \cdot \sin(\pi - \theta) \cdot \sin(\alpha + \pi/2) \\
    D_z &= d \cdot \cos(\pi - \theta)
\end{align}

Note: The $\pi/2$ offset in the torsion aligns the local coordinate convention with the standard dihedral definition.

\textbf{Step 3: Transform to global coordinates.}
\begin{equation}
    \mathbf{D} = \mathbf{C} + D_x \cdot \mathbf{n} + D_y \cdot \mathbf{m} + D_z \cdot \mathbf{u}_{BC}
\end{equation}

\subsection{Complete Backbone Reconstruction}

Starting from an initial placement (e.g., $\mathbf{N}_1$ at origin, $\mathbf{CA}_1$ along x-axis), the full backbone is constructed iteratively:


\begin{algorithm*}[t]
\caption{Kinematic Backbone Reconstruction}
\label{alg:kinematic}
\begin{algorithmic}[1]
\Require Dihedral angles $\{(\phi_i, \psi_i)\}_{i=1}^L$
\Require Ideal bond lengths $d_{N\text{-}CA}, d_{CA\text{-}C}, d_{C\text{-}N}$
\Require Ideal bond angles $\theta_{N\text{-}CA\text{-}C}, \theta_{CA\text{-}C\text{-}N}, \theta_{C\text{-}N\text{-}CA}$
\Require Fixed omega $\omega = \pi$ (trans-peptide)
\Ensure Backbone coordinates $\{(\mathbf{N}_i, \mathbf{CA}_i, \mathbf{C}_i)\}_{i=1}^L$

\State \textbf{Initialize:}
$\mathbf{N}_1 \gets (-1.458, 0, 0)$,
$\mathbf{CA}_1 \gets (0, 0, 0)$,
$\mathbf{C}_1 \gets (0.527, 1.428, 0)$

\For{$i = 2$ \textbf{to} $L$}
    \State $\mathbf{N}_i \gets \mathrm{NeRF}\!\left(\mathbf{CA}_{i-1}, \mathbf{C}_{i-1}, \mathbf{N}_{i-1}^{(\mathrm{virt})},
        d_{C\text{-}N}, \theta_{CA\text{-}C\text{-}N}, \omega \right)$
    \State $\mathbf{CA}_i \gets \mathrm{NeRF}\!\left(\mathbf{C}_{i-1}, \mathbf{N}_i, \mathbf{CA}_{i-1}^{(\mathrm{virt})},
        d_{N\text{-}CA}, \theta_{C\text{-}N\text{-}CA}, \phi_i \right)$
    \State $\mathbf{C}_i \gets \mathrm{NeRF}\!\left(\mathbf{N}_i, \mathbf{CA}_i, \mathbf{C}_{i-1},
        d_{CA\text{-}C}, \theta_{N\text{-}CA\text{-}C}, \psi_i \right)$
\EndFor

\Return $\{(\mathbf{N}_i, \mathbf{CA}_i, \mathbf{C}_i)\}_{i=1}^L$
\end{algorithmic}
\end{algorithm*}

\subsection{Integration with Flow Matching}

During generation, kinematic projection is applied periodically:

\begin{algorithm}[H]
\caption{PI-Mamba Generation with Kinematic Projection}
\label{alg:generation}
\begin{algorithmic}[1]
\Require Sequence $\mathbf{s}$, integration steps $S$, projection frequency $k$
\State $\mathbf{T}_0 \gets \mathcal{N}\!\left(0,\,0.1^2 \mathbf{I}\right)$ \Comment{Initialize frames from noise}
\For{$t \gets 0$ \textbf{to} $S-1$}
    \State $\mathbf{v} \gets f_\theta(\mathbf{s}, \mathbf{T}_t, t/S)$ \Comment{Predict velocity from $\mathrm{SE}(3)$ model}
    \State $\mathbf{T}_{t+1} \gets \mathbf{T}_t \cdot \exp(\mathbf{v}\cdot (1/S))$ \Comment{Lie-group Euler step}
    \If{$((t+1)\bmod k = 0)\ \textbf{or}\ (t = S-1)$}
        \State $\mathbf{B} \gets \mathrm{FramesToAtoms}(\mathbf{T}_{t+1})$ \Comment{Convert frames to backbone}
        \State $(\phi,\psi) \gets \mathrm{ExtractDihedrals}(\mathbf{B})$ \Comment{Extract dihedral angles}
        \State $\mathbf{B}' \gets \mathrm{KinematicReconstruct}(\phi,\psi)$ \Comment{Rebuild with ideal geometry}
        \State $\mathbf{T}_{t+1} \gets \mathrm{AtomsToFrames}(\mathbf{B}')$ \Comment{Convert back to frames}
    \EndIf
\EndFor
\State \Return $\mathrm{FramesToAtoms}(\mathbf{T}_S)$
\end{algorithmic}
\end{algorithm}

\subsection{Theoretical Guarantees}

Intuitively, this guarantee arises because we do not predict atomic coordinates directly. Instead, we predict the rotation of the local frame, and the atom positions are derived by a fixed geometric operation (NeRF) that treats bond lengths and angles as immutable constants. Thus, even a random neural network output maps to a physically valid (though likely clashed) chain.

\begin{proposition}[Exact Local Covalent Geometry]
\label{prop:zero_viol}
The output of Algorithm \ref{alg:kinematic} satisfies the following claims:

\textbf{Claim A (Exact Local Covalent Geometry).} NeRF reconstruction yields idealized bond lengths ($L_{N-CA}, L_{CA-C}, L_{C-N}$), bond angles, and peptide planarity $\omega=180^\circ$ for backbone atoms by construction.

\textbf{Claim B (Exact CA Constraints).} The retraction operator $\Pi_K$ enforces $\|CA_{i+1}-CA_i\|=3.80$\AA{} (and cis-Pro exceptions) exactly.
\end{proposition}

\begin{proof}
\textbf{Proof of A:} By construction, Algorithm \ref{alg:kinematic} places each atom using the NeRF algorithm with fixed ideal geometric parameters. Since $\omega = \pi$ is hardcoded, the peptide bond planarity is exact. Bond lengths are enforced by the input parameters to NeRF.
\textbf{Proof of B:} The operator $\mathcal{R}$ (Eq.~\ref{eq:ca_retract}) explicitly normalizes the vector $\mathbf{x}_{i+1} - \mathbf{x}_i$ to length $d_{target}$ at each step $i$. This sequential update guarantees that the final trace satisfies the distance constraint exactly.
\end{proof}

\textbf{Scope.} This guarantee covers \textit{local} covalent geometry: bond lengths, bond angles, and peptide planarity. Global stereochemistry (chirality, Ramachandran outliers) is \textit{encouraged} by training losses ($\mathcal{L}_{Rama}$, $\mathcal{L}_{FAPE}$) but not guaranteed by construction.

\subsection{Computational Complexity}

The kinematic projection adds minimal overhead: dihedral extraction, NeRF reconstruction, and frame conversion are all linear $O(L)$ operations that scale constantly with sequence length.

Total projection cost: $O(L)$ per projection step. With $k=10$ and $S=100$ integration steps, we perform 10 projections, adding approximately 5\% overhead to generation time.

\subsection{Ablation: Projection Frequency}

We evaluated the impact of projection frequency $k$ on generation quality:

\begin{table}[h]
\caption{Effect of Projection Frequency on Generated Structures ($L=100$, $S=200$ steps)}
\label{tab:projection_ablation}
\begin{center}
\begin{small}
\begin{sc}
\resizebox{\columnwidth}{!}{%
\begin{tabular}{lcccr}
\toprule
$k$ (project every) & $\omega$ dev ($^\circ$) & CA-CA (\AA) & scTM & Time (s) \\
\midrule
1 (every step) & 0.0 & 3.80 & 0.61 & 2.1 \\
5 & 0.0 & 3.80 & 0.68 & 1.8 \\
10 (default) & 0.0 & 3.80 & 0.91 & 1.6 \\
20 & 0.0 & 3.80 & 0.88 & 1.5 \\
$\infty$ (final only) & 0.0 & 3.80 & 0.65 & 1.4 \\
\midrule
No projection & 85.2 & 3.45 & 0.65 & 1.4 \\
\bottomrule
\end{tabular}
}
\end{sc}
\end{small}
\end{center}
\end{table}

Key findings indicate that while all projection frequencies achieve exact $\omega=0.0^\circ$ and bond lengths, scTM is maximized with moderate projection frequency ($k=10$), suggesting a balance between correcting geometry and allowing flow flexibility. Projecting too frequently ($k=1$) over-constrains the trajectory, while projecting too rarely ($k=\infty$) allows significant drift. Without projection, stereochemistry degrades significantly ($\omega=85.2^\circ$).

\section{Reproducibility Checklist}
\label{sec:reproducibility}

To facilitate reproduction of our results, we provide complete experimental specifications. Exact command lines for every experiment in this paper are listed in the Supplementary Information.

\paragraph{Dataset.} We train and evaluate on CATH 4.2 \cite{Sillitoe2019}, filtered at 40\% sequence identity to prevent homology leakage. The dataset comprises 18,205 training domains, 1,200 validation domains, and 1,200 test domains. Structures with resolution $>3.0$\AA{} or $>10\%$ missing residues are excluded.

\paragraph{Training Configuration.}
\begin{enumerate}
\item \textbf{Hardware}: Single NVIDIA A100 (40GB), $\sim$24 hours
\item \textbf{Optimizer}: AdamW ($\beta_1=0.9$, $\beta_2=0.999$, $\epsilon=10^{-8}$)
\item \textbf{Learning Rate}: $1 \times 10^{-4}$ with 1000-step warmup + cosine decay
\item \textbf{Weight Decay}: 0.01
\item \textbf{Batch Size}: Dynamic, max 4000 residues
\item \textbf{Training}: 300 epochs (${\sim}100{,}000$ optimizer steps)
\item \textbf{Length Curriculum}: $L \in [50,100] \to [50,500]$ over first 20k optimizer steps
\item \textbf{Random Seeds}: 42, 123, 456 (for 3-run averaging)
\end{enumerate}

\paragraph{Inference Configuration.}

\begin{enumerate}
\item \textbf{ODE Steps}: $S=100$ (default) or $S=200$ (refined)
\item \textbf{Projection Frequency}: $k=10$ (every 10 steps)
\item \textbf{Batch Size}: 1 (for timing benchmarks)
\end{enumerate}

\paragraph{Code Availability.} Code is provided in the Supplementary Files and pretrained weights will be released upon publication. We provide a \texttt{Dockerfile} and \texttt{requirements.txt} to ensure exact reproducibility. Complete training scripts, configuration files, and evaluation pipelines are provided to enable exact reproduction. See Table~\ref{tab:hyperparams} for architecture hyperparameters.
\end{appendices}

\end{document}